\newcommand{\be}{\begin{eqnarray}}
\newcommand{\ee}{\end{eqnarray}}
\newcommand{\ba}{\begin{eqnarray*}}
\newcommand{\ea}{\end{eqnarray*}}

\RequirePackage{etoolbox}
\csdef{input@path}{%
 {sty/}
 {img/}
}%
\csgdef{bibdir}{bib/}

\documentclass[12pt]{article}

\usepackage{graphicx} 
\usepackage{booktabs} 
\usepackage{amsmath}
\usepackage{mathtools}
\usepackage{bm}
\usepackage[round,sort,numbers,authoryear]{natbib}

\usepackage{tikz}
\usetikzlibrary{fit,positioning}
\usetikzlibrary{backgrounds}
\usepackage{color}
\usepackage{graphicx}
\usepackage{caption}
\usepackage{subcaption}
\usepackage{float}
\usepackage{url}
\usepackage[colorlinks=true,allcolors=blue]{hyperref}
\usepackage{enumitem}

\usepackage{authblk}

\newtheorem{theorem}{Theorem}[section]
\newtheorem{fact}[theorem]{Fact}

\newtheorem{definition}{Definition}[section]

\newenvironment{proof}[1][Proof]{\begin{trivlist}
\item[\hskip \labelsep {\bfseries #1}]}{\end{trivlist}}

\newcommand{\qed}{\nobreak \ifvmode \relax \else
      \ifdim\lastskip<1.5em \hskip-\lastskip
      \hskip1.5em plus0em minus0.5em \fi \nobreak
      \vrule height0.75em width0.5em depth0.25em\fi}

\usepackage[toc,page]{appendix}

\title{A Nonparametric Bayesian Method for Clustering of High-dimensional Mixed Dataset} 

\author{Chetkar Jha  \thanks{
\footnotesize Department of Statistics, University of Missouri-Columbia,
\footnotesize \texttt{cjfff@mail.missouri.edu}}  }

\date{\today} 
\begin{document}
\bibliographystyle{plainnat}

\maketitle

\begin{abstract}
\textbf{Motivation}: Advances in next-generation sequencing (NGS) methods  have enabled researchers and agencies to collect a wide variety of sequencing data across multiple platforms. The motivation behind such an exercise is to analyze these datasets jointly, in order to gain insights into disease prognosis, treatment, and cure. Clustering of such datasets, can provide much needed insight into biological associations. However, the differing scale, and the heterogeneity of the mixed dataset is hurdle for such analyses. 

\textbf{Results}: The paper proposes a nonparameteric Bayesian approach called \emph{Gen-VariScan} for biclustering of high-dimensional mixed data. Generalized Linear Models (GLM), and latent variable approaches are utilized to integrate mixed dataset. 
Sparsity inducing property of Poisson Dirichlet Process (PDP) is used to identify a lower dimensional structure of mixed covariates. 
We apply our method to Glioblastoma Multiforme (GBM) cancer dataset. We show that cluster detection is aposteriori consistent, as number of covariates and subject grows. As a byproduct, we derive a working value approach to perform beta regression.
\end{abstract}



\section{Introduction}\label{S:intro}
Technological advances in next-generation sequencing (NGS) methods have facilitated researchers to sequence DNA at an unprecedented speed. This has led to proliferation of high resolution genomic data such as transcriptomic data (e.g mRNA expression), epigenomic data (e.g DNA methylation) etc. For instance, The Cancer Genome Atlas (TCGA) has compiled large genomic databases for different tumor types. These databases are sourced from multiple genomic platforms on a common set of samples, and are of different data types. Each of these datasets provide a partly independent and complementary view of the genome \citep{hamid09}. 

A problem of interest in integrative genomics is to analyze these datasets jointly. Specifically, there is a considerable interest in studying biological associations across data types. Biological associations play a crucial role in tumor growth, and studying these interactions can provide for a comprehensive understanding of cancer genetics and molecular biology \citep{lock}. Clustering, an unsupervised approach to group objects into clusters which share common pattern, is often used to find biological associations. \citet{eisen98} first, applied clustering methods to find associations among genes in gene expression data. Since then, clustering methods have been employed to find biological associations. Recently, clustering technique has been applied to discover interactions between biomarkers of continuous (gene expression) and categorical data type (DNA copy number alteration / DNA mutation) (\citet{lee08}, \citet{abidin17}). 

An open problem, in this context, is to perform biclustering or two-way clustering of high-dimensional mixed datasets. 
The advantage of biclustering approach over one-dimensional clustering approach is that clustering is done simultaneously between samples and covariates. In biological context, it means that there may be a group of biomarkers (across data types) that defines the biological process for only a subset of samples \citep{lee13}. In other words, a biclustering approach borrows strength from  local interactions. Several studies have established biological relevance of biclustering methods \citep{biclustrev}. Despite the advantages of biclustering methods, we lack a biclustering method for high-dimensional mixed dataset. In this paper, we propose a biclustering method for integrating high-dimensional mixed dataset with an emphasis on clustering the covariates.

\subsection{Challenges in high dimensional mixed data}
The challenge with high dimensional mixed dataset is two fold. The dataset is heterogeneous and consists of multiple data types (such as continuous, binary etc). It is difficult to combine the information across different data types in a meaningful manner as each data type is on a different scale. Moreover, for biomedical mixed datasets, the number of subjects are relatively small compared to the number of covariates (n $<< $ p). Therefore, one needs to reduce the number of covariates to a smaller number of covariates for subsequent analyses. 

\subsection{Current Approaches and its limitations}
There's an extensive literature for clustering of high dimensional mixed datasets. The popular approaches include  K-means/K-mediods algorithm using Gower's distance and hierarchical clustering approaches. 
Over the years, a number of methods have been developed for clustering of high-dimensional mixed dataset with a focus on clustering the samples. For a detailed review, see \citet{litrev2} and \citet{huang17}. Broadly, these methods could be classified 
as i) Matrix Factorization method, and ii) Model based clustering method.

\emph{Matrix Factorization}: Matrix factorization based approaches makes use of the fact that the data matrix can be written in terms of sparse latent factor matrix. \citet{shen} developed a latent factor approach for integrating multiple datasets of different data types. 
\cite{mo} generalized \citet{shen}'s approach for binary, multicategory, and continuous data types. \citet{JIVE} proposed a latent factor approach (JIVE), where they decomposed the total variation into joint variation and individual variation. 
The limitations of above approaches are that they require normalization across datasets, and assume a linear mapping between the data points and latent factors.

\emph{Model based clustering}: Model based clustering  approaches relaxes the linearity assumption between the latent factors and data point. They incorporate likelihood of data points, and use non-parametric Bayesian method to cluster data points into different groups. 
\citet{savage} proposed a modified version of Hierarchical Dirichlet Process (HDP) \citep{teh05} to jointly model gene-expression and transcription factor binding data. 
\citet{kirk} proposed a general approach to integrate multiple data types, simultaneously. 
\citet{lock} proposed a Bayesian consensus clustering approach. They made use of finite Dirichlet Mixture Models to model each dataset separately. The limitation of above approaches are that it is computationally difficult and expensive for non-normal likelihoods, or where conjugacy can't be easily exploited.

Existing methods in integrative genomics have only been used to identify the subtype or cluster among the patients (or subjects). The other limitation of existing methods is that they do not allow for biclustering (simultaneous) clustering of samples and covariates. With an aim to overcome above limitations, we propose a nested partition model for high-dimensional mixed dataset. 


\subsection{Nested Partition Models} Product Partition models were first studied by \citet{hartigan}. \citet{quintana03} and \citet{quintana06} defined product partition models in a nonparametric Bayesian set up. Let $\{x_{ij} \}, i =1, \cdots, n, j = 1, \cdots, p$ be a data matrix, where each row denotes a sample and each column denotes a covariate. Then, the product partition model in a nonparametric Bayesian set up, can be given as below.
\be\label{partition.model}
f(x_{11}, \cdots x_{np}| \theta_{11} \cdots \theta_{np}, \psi)  &=& \Pi_{i=1}^n \Pi_{j=1}^p p(x_{ij} | \theta_{ij}, \psi), \nonumber \\
\theta_{11},\cdots \theta_{np} | G, \psi &\sim& G(. | \psi ), \nonumber \\
 G &\sim& RPM( M G_0 ),
\ee
where $\theta_{ij}$ is a parameter, $\psi$ is a hyperparameter,  M is concentration parameter, G is a mixing distribution, and RPM is a random probability measure. 

The mixing distribution G is almost surely discrete which facilitates the observations to group into clusters and share common $\theta_{ij}$'s. This allows one to naturally infer the number of clusters. Also, the random partition model is exchangeable under the permutation of cluster indices \citep{muller11}.


The product partition models could be generalized to nested partition models. Recently, many authors have proposed such models (for e.g \citet{rodriguez08}, \citet{rodriguez12}, \citet{lee13}). In nested partition models, the data matrix ($\{x_{ij}\}, i =1, \cdots, n, j=1, \cdots, p$ ) are clustered at two levels, which allows the model to borrow strength from local interactions. 
\citet{lee13} proposed a nonparameteric Bayesian model for clustering of RRPA data. They proposed nested clustering of covariates using Dirichlet Process (DP) both at column and row level. \citet{guhaveera} argued that nested Dirichlet Process (DP) with Poisson Dirichlet Process (PDP) leads to more flexible clustering. They applied their method to gene expression data. \citet{xu13} modified \citet{lee13}'s biclustering approach for histone modification data. Reverse Phase Protein Array data, gene expression data, and histone modification data are continuous, continuous, and count data respectively. A limitation of above methods is that they can't be applied to a mixed dataset, consisting of continuous and categorical data type. We fill this gap and extend \citet{guhaveera}'s approach to mixed datasets.

\emph{Gen-VariScan} is a clustering method which biclusters (simultaneously clusters) mixed data matrix $\{ x_{ij} \},i = 1, \cdots, n, j = 1, \cdots p$. The method can be described in two steps, namely:- i) regression step, ii) clustering step. 
We start with an initial cluster label for every element in the data matrix.
In regression step, we perform regression analysis for every column vector where the dependent variable is original column vector and the independent variables are cluster label vector. The regression method depends on data type of the column vector, see Figure \ref{F:1}. The coefficient of cluster label vector is a latent vector, which is continuous. 
In clustering step, the latent vectors are clustered to q ($ q < p $) PDP clusters. Subsequently, the unique elements of all latent vectors ($\theta_{ik}, i = 1\cdots n, k=1 \cdots q) $ are clustered using DP. 
It is worth noting that one can easily cluster the continuous latent vector, in comparison to, the original column vector (which can be continuous/ discrete). Here, the biclustering of (continuous) latent vectors is used as proxy for biclustering of mixed dataset. 
 Moreover, for latent variable approaches (\citet{chib}) and GLM (\citet{glm}), the latent continuous vector
is approximately normal. This allows us to use conjugacy properties and gain computational efficiency (see Section \ref{S:model}).
The column intercept is used to center each covariate column (across data type). 
Figure \ref{F:2} illustrates our method. There are 4 datasets of different data types namely:- mutation, copy number alteration, methylation, and gene expression datasets. In total, there are n= 6 samples, and p=10 covariates. Clusters of individual cells are denoted by pattern and column level clusters are encoded by colors.

\begin{figure}[!htb]
\centering
\begin{tikzpicture}
\tikzstyle{main}=[rectangle, minimum size = 20mm, thick, draw =black!100,rounded corners, node distance = 5mm]
\tikzstyle{inner0}=[circle, minimum size = 20mm, thick, draw =black!100, node distance = 15mm]
\tikzstyle{inner}=[circle, minimum size = 20mm, thick, draw =black!100, node distance = 15mm]
\tikzstyle{connect}=[-latex, thick]
\tikzstyle{dotted}=[-latex, ultra thin, dash dot]
\node[main, fill = black!10] (c2) [label={[align = center] center:\tiny{\textbf{Linear Reg } } \\ $x_{ij}$  }] { };
\node[main, fill = black!10] (c1) [right=of c2,label={[align = center] center:\tiny{\textbf{Binary Reg} } \\$x_{ij}$  }] { };
\node[main, fill = black!10] (c3) [left=of c2,label={[align = center] center:\tiny{ \textbf{Ordinal Reg} }\\ $x_{ij}$  }] { };
\node[main, fill = black!10] (c4) [right=of c1,label={[align = center] center:\tiny{\textbf{Beta Reg} }\\ $x_{ij}$  }] { };
\node[main, fill = black!10] (c5) [left=of c3,label={[align = center] center:\tiny{ \textbf{Poisson Reg} }\\ $x_{ij}$  }] { };
\node[inner, fill = black!10] (b1) [below =of c1,label={[align = center] center:\tiny{\textbf{Intercept}} }] { };
\node[inner, fill = black!10] (b2) [below =of c2,label={[align = center] center:\tiny{\textbf{Intercept} }  }] { };
\node[inner, fill = black!10] (b3) [below  =of c3,label={[align = center] center:\tiny{\textbf{Intercept}}  }] { };
\node[inner, fill = black!10] (b4) [below =of c4,label={[align = center] center:\tiny{\textbf{Intercept}}   }] { };
\node[inner, fill = black!10] (b5) [below =of c5,label={[align = center] center:\tiny{\textbf{Intercept}}  }] { };
\node[inner0, fill = black!10] (th) [below left = of b1,label={[align = center] center:\tiny{\textbf{Coefficient}} \\$\theta_{ik}$ } ] { };
\path(th) edge [dotted] (c1)
(th) edge [dotted](c2)
(th) edge [dotted](c3)
(th) edge [dotted](c4)
(th) edge [dotted](c5)
(b1) edge [dotted](c1)
(b2) edge [dotted](c2)
(b3) edge [dotted](c3)
(b4) edge [dotted](c4)
(b5) edge [dotted](c5);
\end{tikzpicture}
\caption{A visual of data integration in \emph{Gen-VariScan}.}
\label{F:1}
\end{figure}
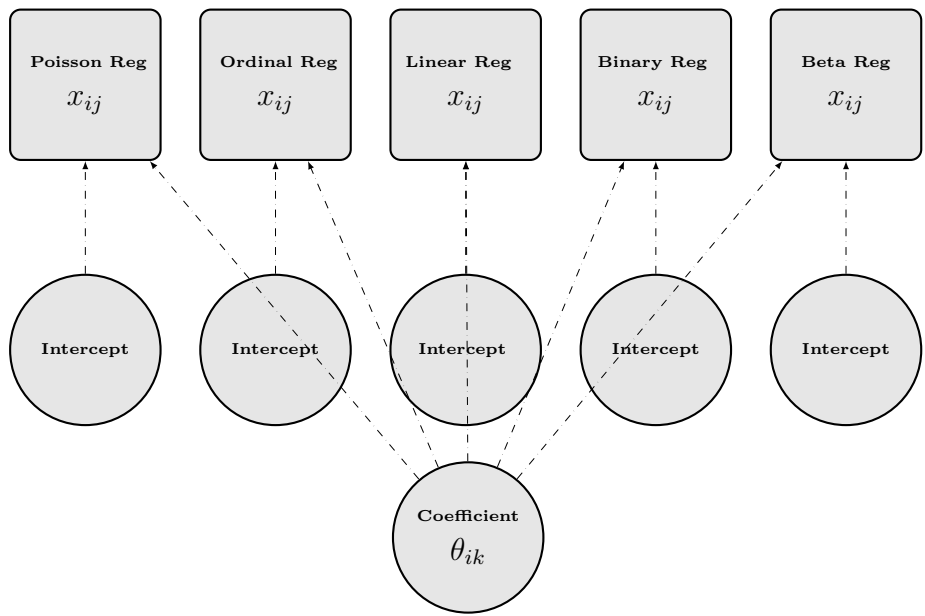
\begin{figure}[!htb]
\centering
\includegraphics[width=15cm]{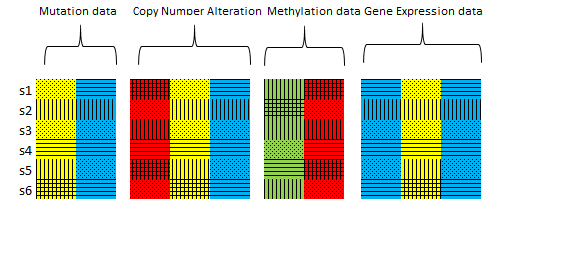}
\caption{An illustration of \emph{Gen-VariScan}. Each row corresponds to sample and each column corresponds to covariate. The figure shows that covariate clusters can be shared across different datasets, which are of different data types.}
\label{F:2}
\end{figure}

\emph{Gen-VariScan} has following advantages over other integrative approaches. It doesn't normalize the datasets and therefore can effectively capture heterogeneity in the datasets. It borrows strength from local clustering, which leads to more flexible clustering. It is computationally efficient  as it exploits the conjugacy between the density of latent variable and the base distribution of DP. The main contribution of this paper are as follows:- 1) We extend \citet{guhaveera}'s approach to mixed datasets, 2) We prove that cluster allocation is aposteriori consistent, 3) As a byproduct, we derive a working value approach for beta regression.

We have drawn motivation for our model from high through-put biomedical set up but our method can be used elsewhere to perform i) data integration, ii) biclustering (simultaneous clustering) of mixed datasets, iii) dimension reduction. The rest of the paper is organized as follows. In Section \ref{S:beta.reg}, we propose a working value approach to do beta regression which is used subsequently in biclustering of mixed dataset. In Section \ref{S:model}, we describe our model. In Section \ref{S:post.inf}, we describe posterior inference for implementing our model. In Section \ref{S:clust.consistency}, we discuss methods to evaluate our model. In Section \ref{S:simulation}, we perform simulation analyses of the model. In Section \ref{S:data}, we apply the model on a real dataset and discuss the results. Section \ref{S:conclusion} summarizes our paper. Supplementary materials contain the theorem proofs, as well as additional data analysis results.

\section{Beta Regression}\label{S:beta.reg}
Beta regression has received considerable attention in recent years. \citet{ferrari} reparametrized  beta density and proposed a classic beta regression model for constant dispersion parameter. \citet{branscum} proposed Bayesian beta regression. \citet{simas} proposed a generalized model where the dispersion parameter is not constant. \citet{mixedbeta} proposed mixed beta regression. The limitation with above approaches is that they can't be readily used in an ensemble set up (such as ours), where one borrows strength from multiple models. A workaround would be to find a working value approach for beta regression \citep{glm}. Recently, \citet{bayesbeta} proposed a working value based approach for Bayesian beta regression. However, their approach isn't numerically stable. We explain this later in the section. But, first, let us define beta density.
\subsection{Beta Density}
A random variable y, follows a beta distribution, whose probability density function is given as below.
\ba
f(y | \alpha, \beta) = \frac{ \Gamma( \alpha + \beta ) }{ \Gamma(\alpha) \Gamma(\beta)} y^{ \alpha - 1} (1-y)^{\beta - 1}, 
\ea
where $\alpha >0 , \beta > 0$ are shape parameters and $ 0 \le y \le 1$. \citet{ferrari} suggested a reparameterization of the beta density in terms of $\mu$ (the mean parameter) and $\varphi$ (the dispersion parameter).
\ba
\mu= \frac{ \alpha }{ \alpha + \beta },\\
\varphi = \alpha + \beta.
\ea
The reparametrized beta density is given as below.
\be
f(y | \mu, \varphi) = \frac{ \Gamma(\varphi ) }{ \Gamma(\mu\varphi - 1) \Gamma((1-\mu)\varphi)} y^{ \mu\varphi - 1} (1-y)^{ (1-\mu)\varphi - 1}, 
\ee
where $0 < \mu < 1 , \varphi > 0$ are mean and disperson parameter, respectively.

\citet{bayesbeta} made use of working variable approach which is commonly carried out to implement Generalized Linear Models (GLM), for details see \citep{glm}. Their proposed working value is given as below.
\be\label{beta.div}
\tilde{\bf{y}} = \bf{x}^{'} \beta^{(c)} + \frac{ y - \mu^{(c)} }{ (\mu^{(c)} )(1- \mu^{(c)} ) },
\ee
\begin{figure}[!htb]
\centering
\includegraphics[width=12cm, height=6cm]{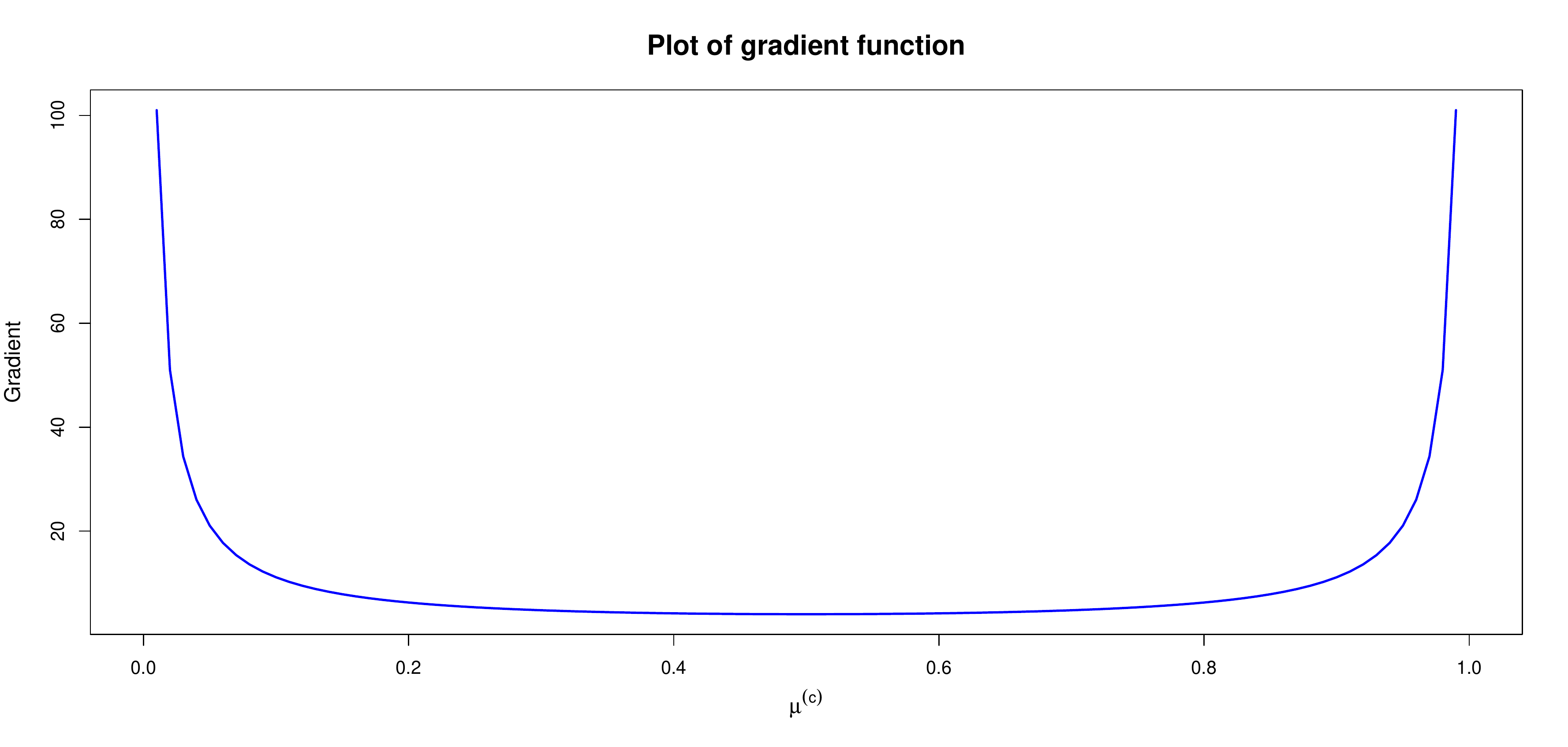}
\caption{Plot of gradient  in equation (\ref{beta.div}) .}
\label{F:3}
\end{figure}
where {\bf{y}} denotes observation, $\tilde{\bf{y}}$ denotes working value, x denotes the covariates for regression, ${\bf{\mu^{(c)}}}$ and ${\bf{\beta^{(c)}}}$ denote the current value of ${\bf{\mu}}$ and ${\bf{\beta}}$ (coefficients) respectively. The limitation of above approach is that it doesn't work well  when ${\bf{\mu^{(c)}}}$ is moderately close to 0 or 1. The gradient diverges and therefore working value approach becomes numerically unstable, see Figure \ref{F:3}. 

\subsection{Proposed Beta Regression}
Let $y_1,\cdots, y_n$ follow beta distribution. Let $x_1,\cdots,x_p$ be the covariates and $\beta_1,\cdots \beta_p$ denotes the coefficients. Also, let $\mu$ denotes the mean of $y_i$ for i = $1 ,\cdots, n$ and, the link function be given as below.
\be
g(\mu) = \Sigma_{i=1}^p x_i \beta_i,
\ee 
where g is the link function. 

The link function g, maps from [0,1] $\to$ R. There are a number of choice for link functions, for example $g(\mu) = F^{-1}(\mu)$, where F is any cumulative distribution function. One may also specify the link function as the complementary log link $g(\mu) = log( - log(1-\mu) )$, the log-log link $g(\mu)= -log(-log(\mu))$ among others. We specify the link function as below.
\be 
g(\mu) = log( \frac{\mu}{1-\mu} ).
\ee
The above link function has a natural interpretation in terms of odds ratio. For detailed discussion on link functions, see  \citet{glm}.
\begin{theorem}\label{GLM.beta}
The score function for the reparametrized beta density is given as
\be
 u_j=\Sigma W (y^{\star} - \mu^{\star} ) \frac{ d \eta}{d \mu^{\star} } x_j.
\ee
Furthermore, the score function satisfies following regularity conditions :-
\be\label{reg.cond}
u_j=\Sigma W (y^{\star} - \mu^{\star} ) \frac{ d \eta}{d \mu^{\star} } x_j &=& 0, \nonumber \\
A \delta b &=& u,
\ee
where b denotes the current estimate of $\beta$ and $\delta b$ denotes the adjustment. \\
Then, the working value for beta regression is given as
\be\label{wkng.val}
z = \eta + \frac{ (y^{*} - \mu^{*} )}{ \varphi*( \psi^{'}( \mu \varphi ) + \psi^{'}( (1-\mu) \varphi )  )*(\mu*(1-\mu) ) },
\ee
where $y^{\star} = log( \frac{y}{1-y})$, $\mu^{*} = \psi( \mu \varphi) - \psi( (1-\mu)\varphi )$ , $\psi$ denotes digamma function, $\psi^{'}$ denotes trigamma function, and A is the fisher's information matrix  ($A_{jk} = - E( \frac{du_j}{d\beta_k} ))$.
\end{theorem}

\begin{proof}
In Appendix A.1 .
\end{proof}

\begin{figure}[h]
\centering
\includegraphics[width=12cm,height=6cm]{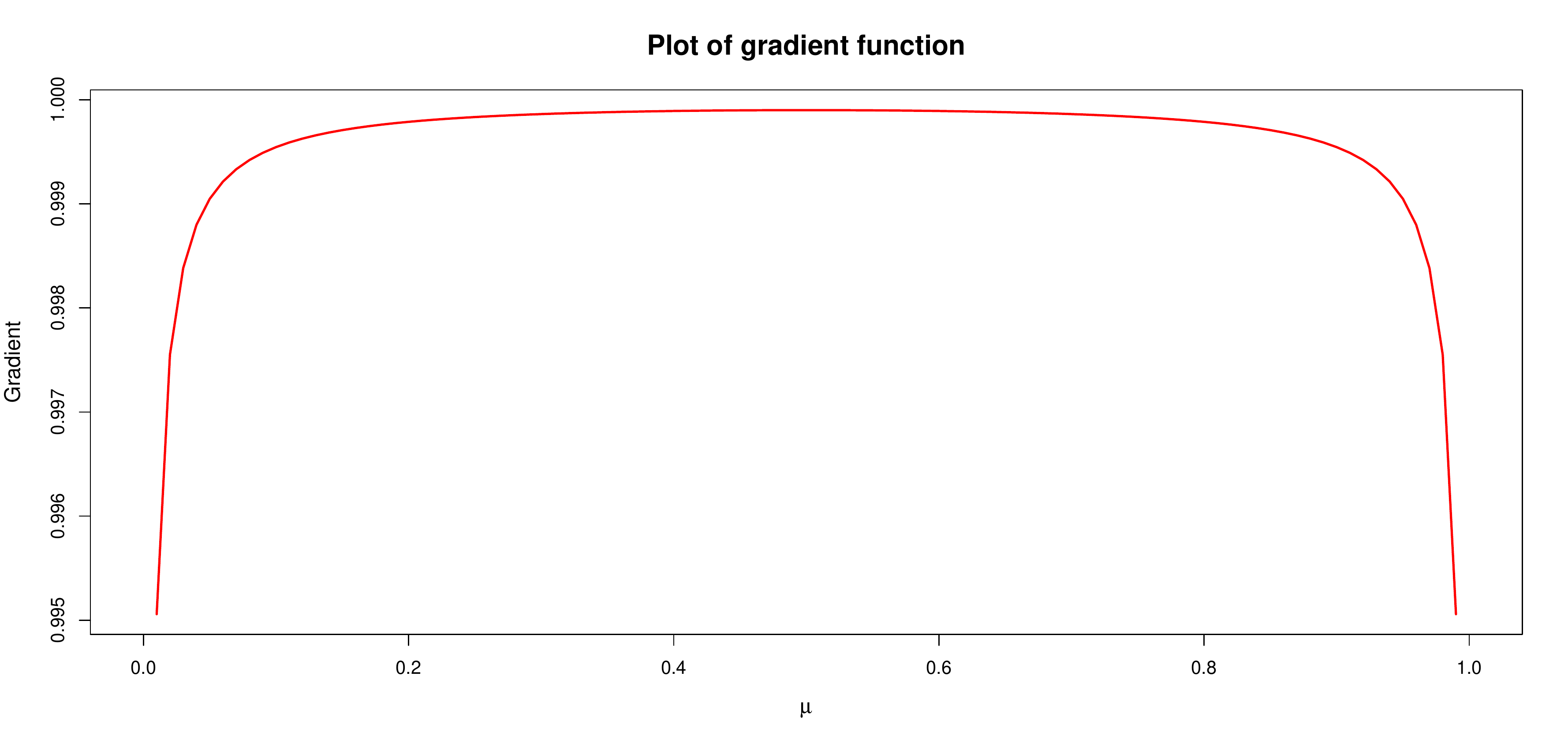}
\caption{The plot of gradient in equation (\ref{wkng.val}) when $\varphi=10,000$.}
\label{F:4}
\end{figure}

\begin{fact}\label{wkng.grad}
The gradient of working value proposed in Theorem \ref{GLM.beta} is bounded from above, and the upper bound is given as below
\be
\frac{ d \eta} {d \mu^*} =\frac{1}{ \varphi*( \psi^{'}( \mu \varphi ) + \psi^{'}( (1-\mu) \varphi )  )*(\mu*(1-\mu) ) } \le \frac{\varphi}{2} < \infty .
\ee
\end{fact}

\begin{proof}
In Appendix A.2.
\end{proof}

Theorem \ref{GLM.beta} gives the working value for beta regression. Figure \ref{F:4} and Fact \ref{wkng.grad} show that the gradient in Theorem \ref{wkng.val} is bounded. Further, the approximate density of z is given as below.
\be
z| y  \sim N( X\beta , \frac{1}{ \varphi^2 ( \mu (1-\mu) )^2 (\psi^{'}(\mu \varphi ) + \psi^{'}( (1-\mu)\varphi ) ) } ) ,
\ee
where $\psi^{'}$ denotes the trigamma function.

In this section, we proposed a working value approach for beta regression. Our working value approach could be used in any other settings to perform Bayesian beta regression.
We use above result in proposing a nonparameteric Bayesian model for mixed datasets, which includes proportion data. 

\section{Model}\label{S:model}
Suppose we take continuous, binary, ordinal, count, and proportion measurements on p biomarkers on n patients. These measurements can be organized in a data matrix \textbf{X}, with n rows and p columns  (n $<<$ p), where each row is a sample and each column a covariate. Further, we assume that each covariate belongs to one of the five data types, namely:-  binary, ordinal, count, proportion, and continuous data type. 

\emph{Gen-VariScan} consists of three steps i) \emph{Allocation Variable} ii) Latent Vectors iii) Data Augmentation. Figure \ref{F:5} gives the visual representation of the model.

\subsection{Allocation Variable}
Let $S = \{ \textbf{x}_1 , \cdots ,\textbf{x}_ p \}$ denote the set of p covariates, where $\textbf{x}_j =\{ x_{j1} \cdots x_{jn}\}$, the $j^{th}$ covariate vector denotes measurement taken on n patients. A  partition of set S yields q $(q < p)$ disjoint subsets: $S_1, \cdots S_q$, such that $\cup_{i=1}^q S_i = S$, and $S_i \cap S_i' = \emptyset, i \neq i'$.  Let $c_j$ denote the cluster membership of $\textbf{x}_j$ to $\{S_1, \cdots, S_q\}$, i.e, $c_j = k$ means $\textbf{x}_j \in k^{th}$ cluster, $j= 1 \cdots p$, $k=1 \cdots q$. We refer to $c_j$ as \emph{allocation variable}, where  j= $1 \cdots p$. 

Following \citet{guhaveera}, we put a two parameter Poisson Dirichlet Process prior (PDP$(M_1, d)$) on allocation variable, where discount parameter $0 \le d < 1$ and precision or mass paramter $M_1 > 0$.
\citet{perman92} introduced PDP and later \citet{pitman95}, \citet{pitman97} studied it further. Further, PDP and DP priors were generalized as Gibbs-type priors in \cite{gnedin05}. 

\citet{guhaveera} gave theoretical and empirical justifications for putting PDP prior. The case for PDP prior can be made as follows. i) DP prior is a specific case of PDP prior (when d =0). ii) A PDP prior has sparsity inducing property which effectively reduces the number of cluster. Asymptotically, the number of cluster for PDP and DP prior can be given as below, see \citet{guhaveera}.
\[
\begin{cases}
M_1 * log(p) & \text{if } d =0, \\
T_{d,M_1}* p^d & \text{if } 0 < d < 1,
\end{cases}
\]
where $T_{d,M_1} > 0$ as $p \to \infty$.

Since, the PDP assumes exchangability therefore the allocation labels are arbitrary. Without loss of generality, we could assign first covariate (${\bf{x}}_1$) into first cluster, i.e, $c_{1} = 1 $. Thereafter, let's say for j =2, $\cdots$ , p covariate we have $q^{j-1}$ number of unique clusters among $c_{1}, c_{2}, \cdots ,c_{j-1}$  where  $k^{th}$ cluster contains $n_k^{j-1}$ number of covariates. Then the conditional probability that $j^{th}$ covariate is assigned to $k^{th}$ cluster is
\be
P(c_j = k | c_1, \cdots ,c_{j-1} ) \propto
\begin{cases} 
n_k^{j-1}  - d & \text{if }  k = 1, \cdots , q^{j-1}, \\
\alpha_1 + q^{j-1}.d & \text{if }  k = q^{j-1} + 1. 
\end{cases}
\ee

The PDP discount parameter d is given the mixture prior $\{ \frac{1}{2} \delta_0 + \frac{1}{2}U(0,1) \}$, where $\delta_0$ denotes 
the point mass at 0. Posterior inferences of d allows us to select between Dirichlet Process and PDP, which allows for flexible clustering.

\subsection{Latent Vectors}
Let $\theta_k$ denote the column vector with elements $(\theta_{1k} , \cdots \theta_{nk} )$, where $k = 1 \cdots q$ and $q$ is number of column clusters.
We put a prior on the latent vectors $ \theta_1, \cdots, \theta_q$, i.e, $ \theta_i \sim G^{(n)} $, where $G^{(n)}$ specifies a distribution in $R^{(n)}$. Following \citet{guhaveera}, we write $G^{(n)}$ as n-fold product  measure of univariate distribution. This essentially imposes a lower dimensional structure on $\theta_k$'s. Further, the prior helps us create nested clustering of subjects within covariates, which can capture local clustering and borrow strength across samples (patients) and covariates (biomarkers).
\be
\theta_{ik} \sim G ,  i =1,\cdots,n  , k = 1, \cdots, q.
\ee
The unknown density is given a Dirichlet Process prior
\be
 G  \sim DP(M_2 G_0 ),
\ee
where the mass parameter M $ > $ 0 and the base distribution $G_0$ as a univariate normal distribution N($\mu_2$, $\tau_2^2$).

\subsection{Data Augmentation}
Posterior sampling in discrete regression models can be computationally challenging. \citet{chib} implemented data augmentation approach for binary and ordinal regression. Bayesian approaches in generalized linear models (GLM) makes use of working value approaches, see \citet{mallick}, \citet{guha08}. Recently, nonparameteric Bayesian methods along with GLM and data augmentation approaches are used for discrete regression models. \citet{atisso} studied nonparametric Bayesian methods for binary regression.
\citet{hannah} used GLM with Dirichlet Process (DP) for regression. \citet{deyoreo} used latent variable approch in nonparametric Bayesian setup for ordinal regression.

We introduce a latent vector vector $\textbf{z}_j$ correspond to every covariate vector $\textbf{x}_j, j = 1 \cdots p$. We define $\textbf{z}_j$  for cases, where $\textbf{x}_j$ belongs to one of the five data types, namely:- binary, ordinal, count, proportion and, continuous. For simplifying notations, we specify $i=1, \cdots, n, j=1, \cdots, p, k=1, \cdots, q, k=1, \cdots, q^{(0)}$, $q$ denotes the estimated number of column cluster and $q^{(0)}$ denotes the true number of column cluster, for all the five data types.

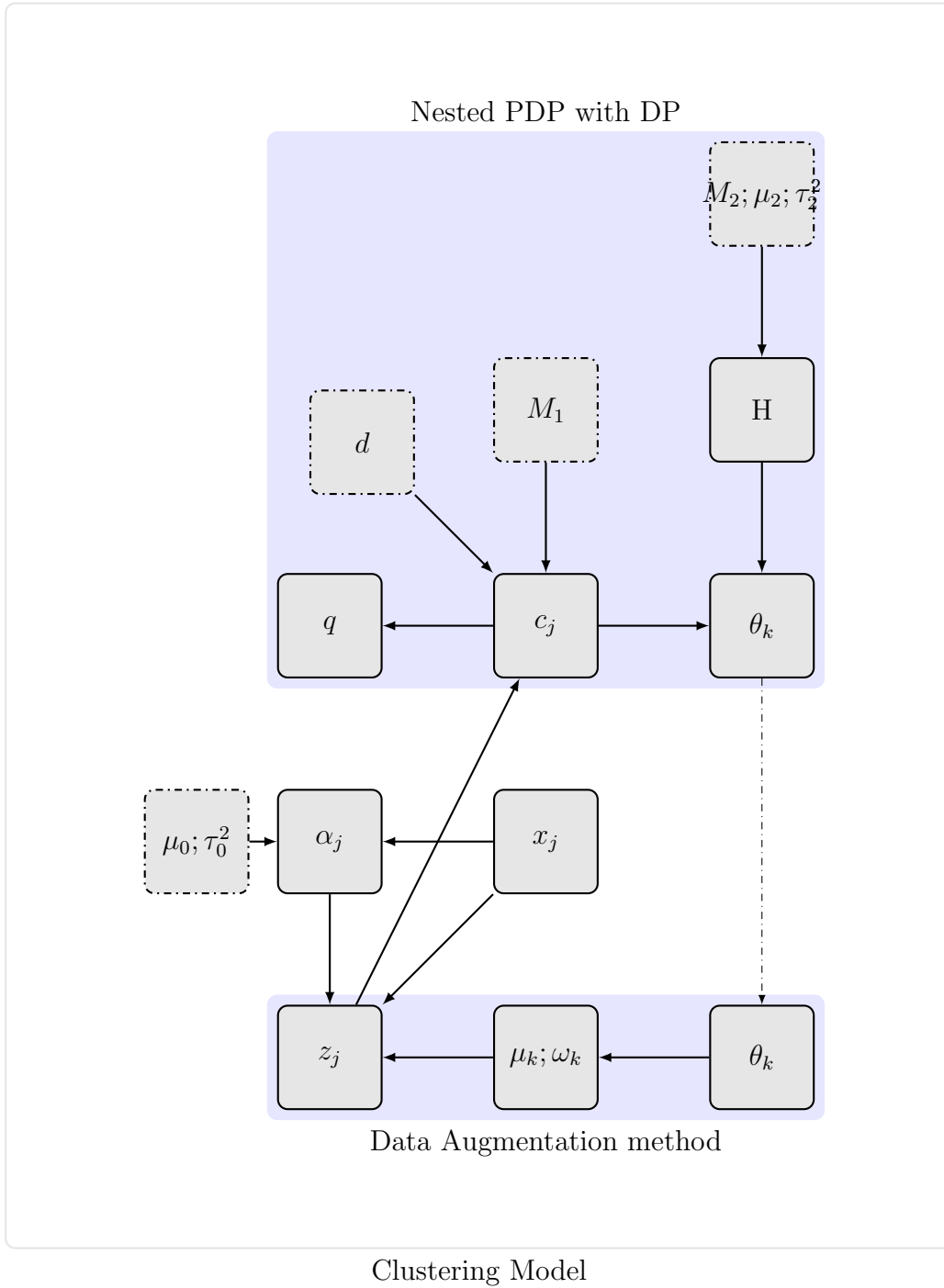
\begin{figure}[H]
\centering
\begin{tikzpicture}
\tikzstyle{main}=[rectangle, minimum size = 15mm, thick, draw =black!100, node distance = 16mm,rounded corners]
\tikzstyle{inner}=[rectangle, minimum size = 15mm, thick, draw =black!100, node distance = 16mm, dash dot,rounded corners]
\tikzstyle{connect}=[-latex, thick]
\tikzstyle{dotted}=[-latex, thin, dash dot]
\tikzstyle{box}=[rectangle, draw=black!100,rounded corners]
\tikzstyle{bigbox1} = [fill=blue!10, draw=blue!10, rounded corners,rounded corners, rectangle]
\tikzstyle{bigbox2} = [fill=blue!10, draw=blue!10 , rounded corners,rounded corners, rectangle]
\tikzstyle{bigbox3} = [inner sep =20mm,draw=black!10, thick, rounded corners, rectangle]
\node[main, fill = black!10] (x) [label=center:$x_{j}$] { };
  \node[main, fill = black!10] (c) [above=of x,label=center:$c_j$] { };
\node[inner, fill=black!10](M)[above= of c, label = center:$M_1$]{};
\node[inner, fill=black!10](d)[above left= of c, label = center:$d$]{};
 \node[main, fill = black!10] (theta) [right=of c,label=center:$\theta_{k}$] { };
\node[main, fill = black!10] (mu) [below=of x,label=center:$\mu_{k};\omega_{k}$ ] { };
 \node[main, fill = black!10] (th) [right=of mu,label=center:$\theta_{k}$] { };
 \node[main, fill = black!10] (H) [above=of theta,label=center:H] { };
  \node[main, fill = black!10] (z) [left=of mu, label=center:$z_{j}$] {};
\node[main, fill = black!10] (alpha) [ left=of x,label=center:$\alpha_{j}$] { };
\node[main, fill = black!10] (q) [left=of c,label=center:$q$] { };
  \node[inner, fill = black!10] (w) [above=of H,label= center:$M_2; \mu_2; \tau_2^2$] { };
    \node[inner, fill = black!10] (prior) [left=of alpha,label= center:$\mu_0; \tau_0^2$,xshift =12mm] { };
\path(d) edge [connect] (c)
      (c) edge[connect](theta)
     (th) edge[connect](mu)
(theta)edge[dotted](th)
 (mu) edge[connect](z)
 (x) edge[connect](z)
(x) edge[connect](alpha)
(alpha) edge[connect](z)
(H) edge[connect](theta)
(c) edge[connect](q)
(z) edge[connect](c)
       (M) edge [connect] (c)
(w) edge[connect](H)
(prior)edge[connect](alpha); 
\begin{pgfonlayer}{background}
  \node[bigbox3] [fit = (theta)(z)(alpha)(prior)(w)(d),label=below :Clustering Model] {};
  \node[bigbox2] [fit = (th) (mu)(z), label = below :Data Augmentation method] {};
  \node[bigbox1] [fit = (theta) (c)(q)(H)(w),label =  above : Nested PDP with DP] {};
\end{pgfonlayer}
\end{tikzpicture}
\caption{A graphical represention of our clustering model. The above representation also highlights smaller blocks of our clustering model namely GLM/latent variable method, and Nested PDP with DP. Squares represents stochastic model parameters, and dotted squares denote model hyperparameters.};
\label{F:5}
\end{figure}

\subsubsection*{Binary Data}

We assume binary datapoint, $x_{ij}$ to have an underlying likelihood, which is given below.
\be\label{binary.def}
h(x_{ij} | c_j^{(0)} = k,\theta^{(0)}_{ik} )= (\Phi( \theta^{(0)}_{ik} ) )^{x_{ij} } (1 - \Phi( \theta^{(0)}_{ik} ) )^{1-x_{ij} },
\ee
where $c_j^{(0)}$ is true cluster allocation, $\theta^{(0)}_{ik}$ denote true latent vector element.

Following \citet{chib}, we generate a latent variable $z_{ij}$, as follows.
\be\label{binary.reg}
g( z_{ij}|(x_{ij}, c_j = k, \alpha_j , \theta_{ik}) ) & = \begin{cases}	
 N( \alpha_j + \theta_{ik}, 1) \text{ truncated at the left by 0},& \text{if } x_{ij} = 1,  \\
 N( \alpha_j + \theta_{ik}, 1) \text{ truncated at the right by 0} ,&\text{if } x_{ij} = 0,
\end{cases}
\ee
where $\alpha_j$ is $j^{th}$ column intercept, $\theta_{ik}$ is estimated latent vector element, $c_j$ is estimated allocation variable.

Equation (\ref{binary.reg}) isn't identifiable. A simple workaround the identifiability problem is to fix $\alpha_j$. Note that we are interested in finding a group of biomarkers, which are similar in the sense of correlation/concordance. Therefore, fixing $\alpha_j$ as column intercept doesn't affect our analysis. A reasonable choice of $\alpha_j$ can be given as below.
\ba
\alpha_j = \frac{\Sigma_{i=1}^n \Phi^{-1}( \frac{x_{ij} }{2 + \epsilon_0 } ) }{n},
\ea
where $\epsilon_0 =0.01$. 

\subsubsection*{Ordinal Data}
 
We assume ordinal datapoint $x_{ij}$, to be distributed as follows.
\be\label{ordinal.def}
h(x_{ij} | \gamma^{(0)}_1 \cdots \gamma^{(0)}_L, c^{(0)}_j = k, \alpha_j,\theta_{ik} )=
\Pi_{l=1}^L (\Phi(\gamma_{l+1} - \theta^{(0)}_{ik} ) - \Phi(\gamma_{l}- \theta^{(0)}_{ik} ) )^{I\{x_{ij} = l \} },  
\ee
where $\gamma^{(0)}_l$'s are true cut offs, $c_j^{(0)}$ is true cluster allocation, $\theta^{(0)}_{ik}$ denote true latent vector element.

Following \citet{chib}, we generate a latent variable $z_{ij}$ as follows.
\be\label{ordinal.reg}
g( z_{ij} |(x_{ij},c_j =k,\alpha_j,\theta_{ik}) ) \sim N( \alpha_{j} + \theta_{ik}, 1)  \text{ trunc at the left (right) by } \gamma_{l-1} (\gamma_l) \text{ if } x_{ij} = l,
\ee
where $\alpha_j$ is $j^{th}$ column intercept, $\theta_{ik}$ is estimated latent vector element, $c_j$ is estimated allocation variable,  and q is estimated number of column cluster and  $\gamma_l$'s are estimated cut offs. 

Like in  binary case, equation (\ref{ordinal.reg}) has identifiability problem.  We approximate $\alpha_j$ as $j^{th}$ column center, which is given as follows.
\ba
\alpha_j = \frac{\Sigma_{i=1}^n \Phi^{-1}( \frac{x_{ij} + \epsilon_0 }{L + 2\epsilon_0 } ) }{n},
\ea
where $\epsilon_0 =0.01$. 

We fix $\gamma_0 = -\infty $, $\gamma_L = \infty $ , and $\gamma_1 = -1$  to address identifiabiity concerns. 
The conditional distribution of $\gamma_l$ is uniform in the interval $ [ max( max( z_{ij} : x_{ij} =l ), \gamma_{l-1} ), min( min(z_{ij}: x_{ij} = l +1), \gamma_{l} ) ]$. 

\subsubsection*{Count Data}


We assume count datapoint $x_{ij}$, to be distributed as follows.
\be\label{count.def}
h(x_{ij} | c^{(0)}_j = k, \theta^{(0)}_{ik} ) = exp( x_{ij} log(\theta^{(0)}_{ik}) -  \nu - log( x_{ij}! ) ) .
\ee
where $c_j^{(0)}$ is true cluster allocation, $\theta^{(0)}_{ik}$ denote true latent vector element, and $q^{(0)}$ indicate the number of true column clusters.

Then for log(.) link function, the adjusted dependent variable and it's approximate density is given as below :-
\be\label{count.reg}
z_{ij}| (x_{ij}, c_j=k,\alpha_j)  &=&  \alpha_j + \theta_{ik}+ \frac{ ( x- e^{ \alpha_j +  \theta_{ik}})}{ e^{ \alpha_j + \theta_{ik}}}, \nonumber \\
g(z_{ij}|(x_{ij}, c_j=k,\alpha_j)) &\sim& N( \alpha_j +\theta_{ik}, e^{-( \alpha_j + \theta_{ik} ) }  ),
\ee
where $\alpha_j$ is $j^{th}$ column intercept, $\theta_{ik}$ is latent vector element, $c_j$ is allocation variable. 

Since, $\alpha_j$ isn't identifiable in equation (\ref{count.reg}), we fix $\alpha_j$ as approximate center of  $j^{th}$ column, we fix $\alpha_j$ as below.
\ba
\alpha_j = log(\Sigma \frac{ x_{ij}}{n} + \epsilon ),
\ea
where $j=1, \cdots, p, \epsilon_0 =0.01$. 

\subsubsection*{Continuous Data}
We assume continuous data point $x_{ij}$, to be distributed as follows.
\be\label{norm.def}
h(x_{ij} | c_j=k,\alpha_j, \theta^{(0)}_{ik},\tau_0) \sim N (\alpha_j + \theta^{(0)}_{ik}, \tau_0^2),
\ee
where $c_j^{(0)}$ is true cluster allocation, $\theta^{(0)}_{ik}$ denote true latent vector element, $\tau_0$ is true standard deviation, and $\alpha_j$ is $j^{th}$ common intercept.

The latent variable $z_{ij}$ is given by $x_{ij}$ whose approximate density is given as under. 
\be\label{normal.reg}
g(z_{ij}| (x_{ij},c_j =k, \alpha_j, \theta_{ik})) = N( \alpha_j +\theta_{ik}, \tau^2  ),
\ee
where $\alpha_j$ is the $j^{th}$ column intercept, $\theta_{ik}$ is estimated latent vector element, $c_j$ is estimated allocation variable,$\tau$ is estimated variance.

In order to address the identifiability problem in equation (\ref{normal.reg}), we constraint that $\alpha_j + \theta_{ik}$ to be 0.
\subsubsection*{Proportion Data}

We assume that  proportion data point $x_{ij}$, to be distributed as follows.
\begin{align}\label{proportion.def}
h(x_{ij} | c^{(0)}_j=k,\alpha_j, \theta^{(0)}_{ik} )=
\frac{ \Gamma( \varphi) } { \Gamma( \frac{e^{\theta^{(0)}_{ik} }}{ 1 + e^{\theta^{(0)}_{ik} } } \varphi ) \Gamma(  \frac{1 }{ 1 + e^{\theta^{(0)}_{ik} } } \varphi ) }x_{ij}^{\frac{e^{\theta^{(0)}_{ik} }}{ 1 + e^{\theta^{(0)}_{ik} } } \varphi} (1- x_{ij} )^{\frac{1 }{ 1 + e^{\theta^{(0)}_{ik} } } \varphi },
\end{align}
where $c_j^{(0)}$ is  true cluster allocation, $\theta^{(0)}_{ik}$ denote true latent vector element, and $\varphi$ is a constant for proportion data.

Under the log odds link function and using Theorem \ref{GLM.beta}, the latent variable ($z_{ij}$) and it's approximate density can be given as follows :-
\be\label{proportion.reg}
z_{ij}|(x_{ij}, c_j = k,\alpha_j , \theta_{ik} ) &=& \alpha_j +\theta_{ik} + \frac{ (y_{ij}^{*} - \mu_{ij}^{*} )}{ \varphi*( \psi^{'}( \mu_{ij} \varphi ) + \psi^{'}( (1-\mu_{ij}) \varphi )  )*(\mu_{ij}*(1-\mu_{ij}) ) },\\
g(z_{ij}|(x_{ij}, c_j=k,\alpha_j, \theta_{ik},\varphi )) &\sim& N( \alpha_j +\theta_{ik},\frac{1}{\varphi^2*( \psi^{'}( \mu_{ij} \varphi ) + \psi^{'}( (1-\mu_{ij}) \varphi )  )*(\mu_{ij}*(1-\mu_{ij}) )^2 }  ), \nonumber
\ee
where $y_{ij}^{\star} = log( \frac{x_{ij} }{1-x_{ij} })$, $\mu_{ij}^{*} = \psi( \mu_{ij} \varphi) - \psi( (1-\mu_{ij})\varphi )$ , $\psi$ denotes digamma function, and $\psi^{'}$ denotes trigamma function, $\alpha_j$ is $j^{th}$ column center, $\varphi$ is estimated dispersion, $c_j$ is estimated cluster allocation, $\theta_{ik}$ is estimated latent vector element, and $q$ indicates estimated number of column clusters.

We specify the following prior on $\varphi$ 
\ba
\pi(\varphi) \sim \Gamma(1,1)
\ea 
where $\Gamma$ refers to Gamma density.  

In order to address identifiability problem in equation (\ref{proportion.reg}), we constraint that $\alpha_j + \theta_{ik}=0$ to make it identifiable.
\section{Posterior Inference}\label{S:post.inf}
Posterior inference is computationally expensive. We apply \citet{guha10}'s data squashing algorithm to speed up computation. We start with an initial configuration of the model parameters. The model parameters are iteratively updated by the MCMC procedure. Broadly, the MCMC procedure can be divided into two steps.
\begin{enumerate}
\item Conditional on the current model parameters, we update the data augmentation variable for each of five data types.
Subsequently, we update the allocation variable, latent vector elements, and hyperparameters until the MCMC chain converges. See Appendix A.2 for further details.
\item  We compute  Monte Carlo estimates to compute the posterior probability of clustering for each pair of covariates. \citet{dahl} proposed a method which uses pairwise probabilities to get the point estimate of binary vector, which is called the \emph{least-square allocation}. At the end of MCMC iteration, we use \citet{dahl} to estimate the \emph{least-square allocation}.
\end{enumerate} 

\section{Clustering Consistency}\label{S:clust.consistency}
Clustering consistency, may seem like a desireable property for clustering procedure, but it is hardly guaranteed. \citet{muller11} notes that the random partition models are exchangeable under permutations and hence, actual clusters and cluster related 
inferences are subject to label switching problem, as in finite mixture models. Clustering under finite mixture models often results in non-identifiablity, and redundancy in clusters (see,  \citet{fruhwirth}). \citet{rousseau} showed that a careful choice of priors yields emptied redundant clusters for over-fitted mixture models. \citet{petralia} defined a repulsive process, which leads to better separated clusters. The general strategy of above solutions is to impose identifiability constraints and detect the true number of clusters. However, above methods fall short of discovering the true allocation of objects into clusters. 

Like finite mixture models, non-parameteric Bayesian models doesn't guarantee cluster consistency, until very recently.
In a surprising result, \citet{guhaveera} proposed a biclustering model for clustering continuous data matrix with n rows and p columns. As n (sample size) and p (covariates) becomes large, they showed that their model can discover true cluster allocation of p covariates. The intuition behind this phenomenon is that as n and p becomes large the n-dimensional objects becomes well separated in $\mathcal{R}^n$ and form identifiable clusters. Theorem \ref{clust.consistency} extends \citet{guhaveera}'s result on clustering consistency to mixed dataset.\\
\textbf{True model.} Let \textbf{$X_{np }$ } be a mixed dataset, where $\textbf{x}_j, j = 1, \cdots, p$ belong to one of the five data types, namely:- binary, continuous, count, ordinal, and proportion. Let $\mathcal{X}_{\delta_{x_{ij}} }$ denote the sample space for $x_{ij}$, where $\delta_{x_{ij}}$ indicate the data type of $x_{ij}, i =1, \cdots, n, j =1, \cdots, p.$  Let  $K_{\delta_{x_{ij}}}(. \mid \theta)$ be the probability density function on the space $\mathcal{X}_{\delta_{x_{ij}}}$ which is defined in Definition \ref{th.def.1}. Further, we make the following assumptions about the covariate generating process.

\begin{enumerate}
\item \label{a} The elements of data matrix, $x_{ij}$, are independent, but not identical realizations, $i = 1, \cdots, n, j =1, \cdots, p.$  from a true mixing distribution $P_0$ convoluted with some exponential family densities, which are given in Definition \ref{th.def.1},
\item \label{b} The true mixing distribution $P_0$ is discrete in $\mathcal{R}$, which implies that $P_0^{(n)}$ the true n-variate mixing distribution is discrete, as well. The discreteness of $P_0^{(n)}$ implies the existence of true cluster allocation variable $ (c^{(0)}_{1}, \cdots, c^{(0)}_{n} )$,
\item \label{c}
\ba
x_{ij} \mid (c^{(0)}_{j} = k)  \sim f_{P_0,ij} = \int K_{\delta_{x_{ij}}}(x_{ij} \mid \theta^{(0)}_{ik} )dP_0,
\ea
where $i = 1, \cdots, n$, $j = 1, \cdots, p, k =1, \cdots, q^{(0)}$, $q^{(0)}$ is the number of true clusters, and  $K_{\delta_{x_{ij}}}$ is given in Definition \ref{th.def.1},
\item \label{d} The atoms of $P_0$ are i.i.d realizations of a univariate normal distribution, $G_0$.
\end{enumerate}

\begin{definition}\label{th.def.1}
\ba
K_{\delta_{x_{ij}}}(x_{ij} \mid \theta) =
\begin{cases}
(\Phi(\theta) )^{x_{ij}} (1 - \Phi(\theta) )^{1-x_{ij}} , &  \delta_{x_{ij}} = 1 \\
\Pi_{m=1}^L (\Phi( \gamma_{m+1} - \theta ) - \Phi( \gamma_{m} - \theta ) )^{ I_{\{x_{ij} = m\}} }  , & \delta_{x_{ij}} = 2 \\
\frac{ \Gamma( \phi) } { \Gamma( \frac{e^{\theta }}{ 1 + e^{\theta } } \phi ) \Gamma(  \frac{1 }{ 1 + e^{\theta } } \phi ) }x_{ij}^{\frac{e^{\theta }}{ 1 + e^{\theta } } \phi} (1- x_{ij} )^{\frac{1 }{ 1 + e^{\theta } } \phi } , & \delta_{x_{ij}} = 3 \\
\phi(  \frac{x_{ij} - \theta }{\tau}) , & \delta_{x_{ij}} = 4 \\
\frac{ \theta^{x_{ij} } }{x_{ij}!} \exp(-\theta ) , & \delta_{x_{ij}} = 5
\end{cases}
\ea
\end{definition}
where, $\Phi$ denotes the cdf of standard normal distribution, $\phi$ denotes the standard normal density, $\gamma_m, m = 1\cdots L$ denote the cut off points for the ordinal data, $( i = 1, \cdots, n, j =1, \cdots, p $).

Let $\mathcal{L} =\{j_1, \cdots j_L \} \subseteq \{1 \cdots p \}$ be a fixed subset of L covariate index. true allocation variable be denoted by $c^{(0)}_{j}$ for j = $1 \cdots p$ and the estimated allocation variable be denoted by $c_{j}$ for j = $1 \cdots p$ , then we compute the allocation accuracy by mean-taxicab distance between $c_j$ and $c^{(0)}_{j}$, which is given below.
\ba
\chi_{\mathcal{L}}(c) = \frac{ \Omega_{j_1 \neq j_2 \in \{1 \cdots p\} } \emph{I}( \emph{I}(c^{(0)}_{j_1} = c^{(0)}_{j_2}) = \emph{I}(c_{j_1} = c_{j_2})  ) }{  {L \choose 2 } }
\ea

A low value of $\chi_{\mathcal{L}}(c)$ (e.g 0) indicates low accuracy and a high value of $\chi_{\mathcal{L}}(c)$ (e.g 1)  indicates high accuracy for estimated allocation vector $\textbf{c}$ for the set $\mathcal{L}$. It is worth noting that $\chi_{\mathcal{L}}(c)$ is invariant of the permutation of the cluster labels.

\begin{theorem}\label{clust.consistency}
Let \textbf{$X_{np }$ } be a mixed dataset, where $\textbf{x}_j, j = 1, \cdots, p$ belong to one of the five data types, namely :- binary, continuous, count, ordinal, and proportion. In addition to the model assumptions (\ref{a})-(\ref{c}), we assume that $f_{P_0,ij}$ in (\ref{c}) is bounded, i = $1, \cdots, n, j=1, \cdots, p$.

Let $\mathcal{L} =\{j_1, \cdots j_L \} \subseteq \{1 \cdots p \}$ be a fixed subset of L covariate index. Then, there exists an increasing sequence $\{ p_n \}$, that grows with n provided $p > p_n$, and the clustering inferences for the covariate subset $\mathcal{L}$ are aposteriori consistent. That is,
\ba
\lim_{{n \to \infty}, {p > p_n} } P[ \chi_{\mathcal{L}}(c) = 1 | X_{np} ] \to 1
\ea
\end{theorem}
\begin{proof}
See Appendix A.4 for the proof.
\end{proof}

\section{Simulation Studies}\label{S:simulation}
To evaluate the performance of \emph{Gen-VariScan} procedure to detect covariate column clusters, we study artificially
simulated datasets under different scenarios. The parameters of true model were chosen to closely match the estimates
of the benchmark data analysis. We also perform a comparison against other methods, when true column clusters are known.

\subsection{Simulation scheme}\label{setup} We investigate the proposed method's accuracy as a clustering procedure for mixed datasets using artificial dataset for which true clustering pattern is known. The artificial data was selected to have the same dimension as benchmark dataset (n =71 subjects and p =352 covariates), and then we compared the co-clustering probabilities of p covariates against truth. The simulation analysis was done for different mix of mixed dataset, including one which closely resembles the mix of benchmark data. Depending on the mix of the data, the artificial data was simulated as below.
\begin{enumerate}
\item \emph{ True Allocation variables:}  We generate $c_1^{(0)}, \cdots, c_p^{(0)}$ as partitions through Poisson Dirichlet Process (PDP) with a discount parameter d = 0.3 and mass parameter $M_1=20$. Thereby, we compute the true number of clusters $Q_0$.
\item \emph{ Latent vector elements :} For i = 1, $\cdots$, n , k = 1, $\cdots, Q_0$ , elements $\theta_{ik}^{(0)} \sim G$, where $G \sim DP(M_2 G_0)$ with mass parameter $M_2$ = 11 and base distribution $G_0= N(\mu_2, \tau_2^2)$ . We choose $\mu_2 =-0.18$, and $\tau_2^2 = 2.2$.
\item \emph{Covariates :} For each covariate column, we assign a data type out of five data types (binary, ordinal, count, continuous, and proportion) using multinomial distribution with a pre-specified probability according to the mix of dataset. 
Subsequently, we simulate the covariate column. For each data type, the covariate column is generated using a density function, which is a function of two components  i) $\alpha_j$  ii) $\theta_{ik}$, where $c_j=k$. Specifically, we generate the $i^{th}$ element of $j^{th}$ covariate column as below ($i = 1, \cdots, n, j=1, \cdots, p, k =1, \cdots, q, l =1, \cdots, 5$).
\begin{itemize}
\item Binary : $ x_{ij}| c_j = k \sim Bernoulli(1, \Phi( \theta_{ik} ) ) .$
\item Continuous : $ x_{ij}| c_j = k \sim N(\alpha_j + \theta_{ik}, \tau^2 ), $
where $\alpha_j =6.27$,  $\tau = 2.14$ .
\item Count : $ x_{ij}| c_j =k \sim Poisson(  exp(\theta_{ik}) )$.
\item Ordinal : The ordinal data points were assumed to have five categories. \\
$
x_{ij} |c_j =k \sim Multinom(1,(p_{ij(1)},p_{ij(2)},p_{ij(3)},p_{ij(4)},p_{ij(5)}), 
$
where $p_{ij(l)} |(c_j = k)=\Phi(\gamma[l+1]) - \theta_{ik}) - \Phi(\gamma[l]) - \theta_{ik}), \gamma = \{ -Inf, -2,-1,0,1, Inf \}.$
\item Proportion : 
$
x_{ij} | (c_j = k) \sim Beta( \mu_{ij}\varphi, (1-\mu_{ij} )\varphi ), 
$

where 
$\mu_{ij}  | (c_j = k) = \frac{ \exp(\theta_{ik})}{ 1 + \exp(\theta_{ik} ) },\varphi = 19.43$.
\end{itemize}
\end{enumerate}

\subsection{Comparison}
The artificial data was simulated as per the settings in section \ref{setup} for seven different scenarios. The scenarios depend on the mix of data type in the dataset. Five scenarios corresponds to five data types, benchmark data mix scenarios resembles mix of data type in the benchmark data, and uniformly mix of data type across five data types. Further, the comparison study was replicated 15 times. For each dataset, we ran the simulation for 15,000 iterations, where we ignored the first 5,000 iterations as burn-in. Subsequently, using \citet{dahl}, we estimated a point estimate for cluster allocations, called the \emph{least-squared configuration}, and denoted by $\hat{c}_1, \cdots, \hat{c}_p$. Finally, we estimated the accuracy of our method by estimating the proportion of correctly predicted clustered covariate pairs, $\hat{\chi}$, which is given below.
\ba
\hat{\chi} = \frac{ \Omega_{j_1 \neq j_2 \in \{1 \cdots p\} } \emph{I}( \emph{I}(c^{0}_{j_1} = c^{0}_{j_2}) = \emph{I}(c_{j_1} = c_{j_2})  ) }{  {p \choose 2 } }.
\ea

\begin{figure}[t]
\centering
\includegraphics[width=14cm,height=10cm]{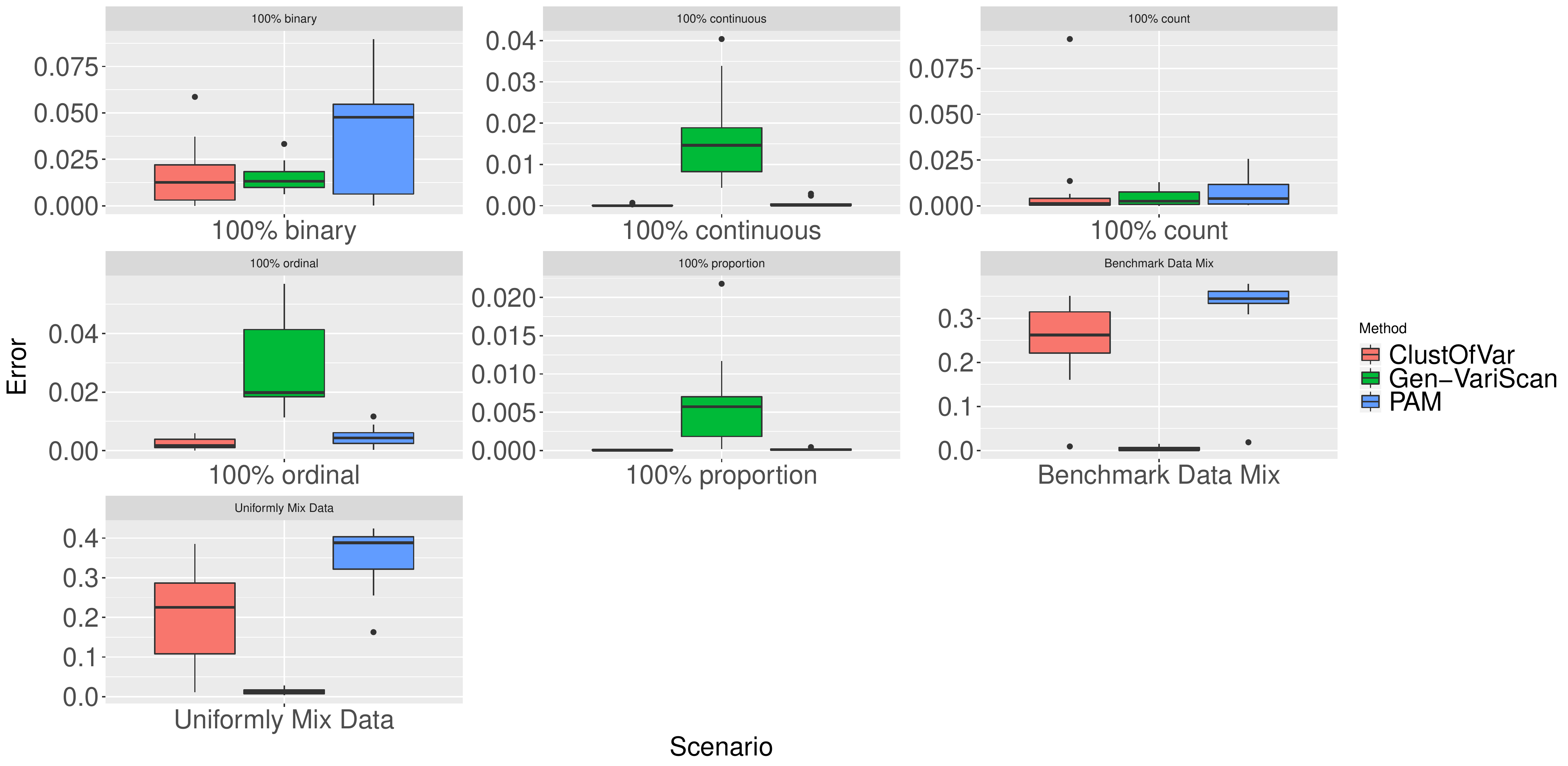}
\caption{The figure compares misclassification error ( 1- $\chi$) of our method (in green) against PAM (in blue) and ClustOfVar (in red). The comparison is made for seven scenarios (100 \% binary, 100\% continuous, 100\% count, 100 \% ordinal, 100 \% proportion, Benchmark Data Mix, Uniformly Mix Data, respectively. }
\label{F:comp}
\end{figure}

A high value of $\hat{\chi}$ (e.g 1) or a low value of $1 -\hat{\chi}$ (e.g 0) indicates high clustering accuracy.
We evaluate our method against Partition around mediods (PAM) with gower's distance \citep{PAM}, and ClustOfVar's hierarchical clustering approach \citep{ClustOfVar}. The number of cluster for PAM and ClustOfVar was estimated using maximum silhoutte width.
For each of the dataset, we ran the MCMC for 15,000 iterations where the first 5,000 iteration was discarded as burn in. Using \citet{dahl} we computed a point estimate for cluster allocations. Subsequently, we evaluated the three method by computing the proportion of incorrectly predicted clustered covariate pairs, $1- \hat{\chi}$.

Figure \ref{F:comp} gives the boxplot of error ( $1 - \hat{\chi}$) for each of the three methods for seven different scenarios. The figure further shows that PAM and ClustOfVar fairs marginally better than \emph{Gen-Variscan} when the covariates consists of a single data type. The accuracy of PAM and ClustOfVar could be attributed to the phenomenon that objects in $\mathcal{R}^n$ tends to become well separated for large n (see Section \ref{S:clust.consistency}). PAM and ClustOfVar can detect these well separated objects when datasets consists of single data type.
However, the figure shows that PAM and ClustOfVar are highly inaccurate for mixed data scenarios whereas \emph{Gen-Variscan} maintains a consistency in accuracy across different mix of data types. This shows that \emph{Gen-Variscan} is well-suited to cluster mixed datsets, in comparison to other two methods.
The boxplot (see Figure \ref{F:comp}) shows that \emph{Gen-VariScan}'s median error ($1 - \hat{\chi}$) of incorrectly classifiying covariate pairs is atmost about 0.02, across all scenarios. This suggests that \emph{Gen-VariScan} correctly classifies about 60,886 covariate pairs out of $\binom{352}{2}$ and incorrectly classifies about 1,242 covariate pair out of total $\binom{352}{2}$.
It is worth pointing that \emph{Gen-Variscan}'s accuracy is consistently across the scenarios. 

Table \ref{T:1} gives the 95 \% credible interval for the lower bound of the Bayes Factor of a PDP model and a DP model.
The lower bound of the Bayes Factor is given by $log( P( d !\neq 0 \mid \textbf{X} )/ P(d = 0 \mid \textbf{X} ) )$. 
Besides 100 \% count and 100 \% binary, we find overwhelming evidence in favor for the PDP model against  DP model.
This is true even though we had a put a prior $ \frac{1}{2} I\{d=0\} + \frac{1}{2}U(0,1)$.
For 100 \% binary scenario, there's a strong evidence in favor of the PDP model, but it's inconclusive in case of 100 \% count scenario.
The table also gives the 95 \% credible interval for the estimated d. We note that the true value of d, i.e 0.3, lie within the 95 \% credible interval for all the scenario. This validates that our model can estimate true level of sparsity among column clusters.

\begin{table}[]
\centering
\caption{ 95 \% Credible Interval}
\begin{tabular}{ | l | l | l | l | l | }
\hline
 & \multicolumn{2}{l |}{Estimated d} & \multicolumn{2}{l |}{ log Bayes Factor}  \\
\hline 
Scenario &  Lower C.I     & Upper C.I     &  Lower C.I      & Upper C.I         \\
\hline
Benchmark Data Mix & 0.200         & 0.401          & Inf          & Inf          \\
Uniformly Mix Data & 0.262         & 0.484          & Inf           & Inf          \\
100\% binary & 0.078           & 0.399         & 0.974          & Inf          \\
100\% continuous & 0.229        & 0.547        & Inf          & Inf         \\
100\% count & 0.00        & 0.32        & -1.77          & Inf         \\
100\% ordinal & 0.274          & 0.386       &  Inf         & Inf         \\
100\% proportion & 0.183          & 0.462         & 6.87          & Inf \\ \hline
\end{tabular}
\label{T:1}
\end{table}

We also evaluate the model for different value of $\tau =\{0.1,0.5,0.9\}$ (standard deviation for continuous data), $\varphi =\{10, 20, 30\}$ (dispersion parameter proportion data) for the benchmark data mix scenario. Table \ref{T:2} gives the 95 \% credible interval for d for benchmark data mix. We see that the true value of d, i.e 0.3,  lies within the credible interval. This reflects that our model can achieve similar level of sparsity as the true model, which further validates our method. 
\begin{table}[]
\centering
\caption{ 95 \% Credible Interval of d}
\begin{tabular}{| l | l | l | l |}
\hline
 \multicolumn{2}{| l |}{ Scenario} & \multicolumn{2}{l |}{95 \% C. l of d} \\
\hline
 $\tau$ &  $\varphi$ &  Lower &  Upper \\
\hline
 0.1& 10  &   0.201        &  0.345          \\
 0.1& 20  &   0.174        &  0.330          \\
 0.1& 30  &   0.175        &  0.332          \\
 0.5 & 10  &  0.292        &  0.439          \\
 0.5 & 20  &  0.148        &  0.311          \\
 0.5 & 30  &  0.284        & 0.436          \\
 0.9 & 10  &  0.160        & 0.317          \\
 0.9 & 20  &  0.298        & 0.450          \\
 0.9 & 30  &  0.241        & 0.374         \\ \hline
\end{tabular}
\label{T:2}
\end{table}

\section{Benchmark Data Analysis}\label{S:data}
\subsection{Data Description}
We downloaded TCGA Glioblastoma Multiforme (GBM) dataset through TCGA2STAT package in R \citep{tcga2stat}, which includes gene expression, copy number alteration, methylation, and  mutation datasets. Each of these datasets is extracted from a different platform. The gene expression dataset comes from Affymetrix Human Genome U133A 2.0 Array, the methylation dataset comes from  Illumina Infinium HumanMethylation27, the copy number alteration comes from Affymetrix SNP6, and the mutation data was obtained from Mutation Annotation Format (MAF). The gene expression and the mutation dataset is at the gene level (level 3), whereas the copy number alteration and methylation dataset was obtained at probe level (level 2) and then mapped to the gene level (level 3).

Each dataset has 10 normal samples. For gene expression and copy number alteration biomarkers, we used Wilcoxon rank sum test to determine whether or not biomarkers were different between the tumor case and the normal sample. A p-value threshold ($p < 10^-7$) was applied to select the biomarkers. We filtered the mutation biomarkers, which had less than 10 number of mutation occurences. Further, we included gene expression, copy number alteration, methylation, and mutation of genes which comprises the RB pathway, RTK/Ras/PI3K/AKT pathway, and TP53 pathway. These pathways are associated with Glioblastoma cancer \citep{TCGA08}. We merged the datasets by patient and kept the patient records which were present in all the four data types. The combined dataset consists of 71 patients and 352 covariates.

\subsection{Results}
\begin{figure}[h]
\centering
\includegraphics[width=12cm, height = 6cm]{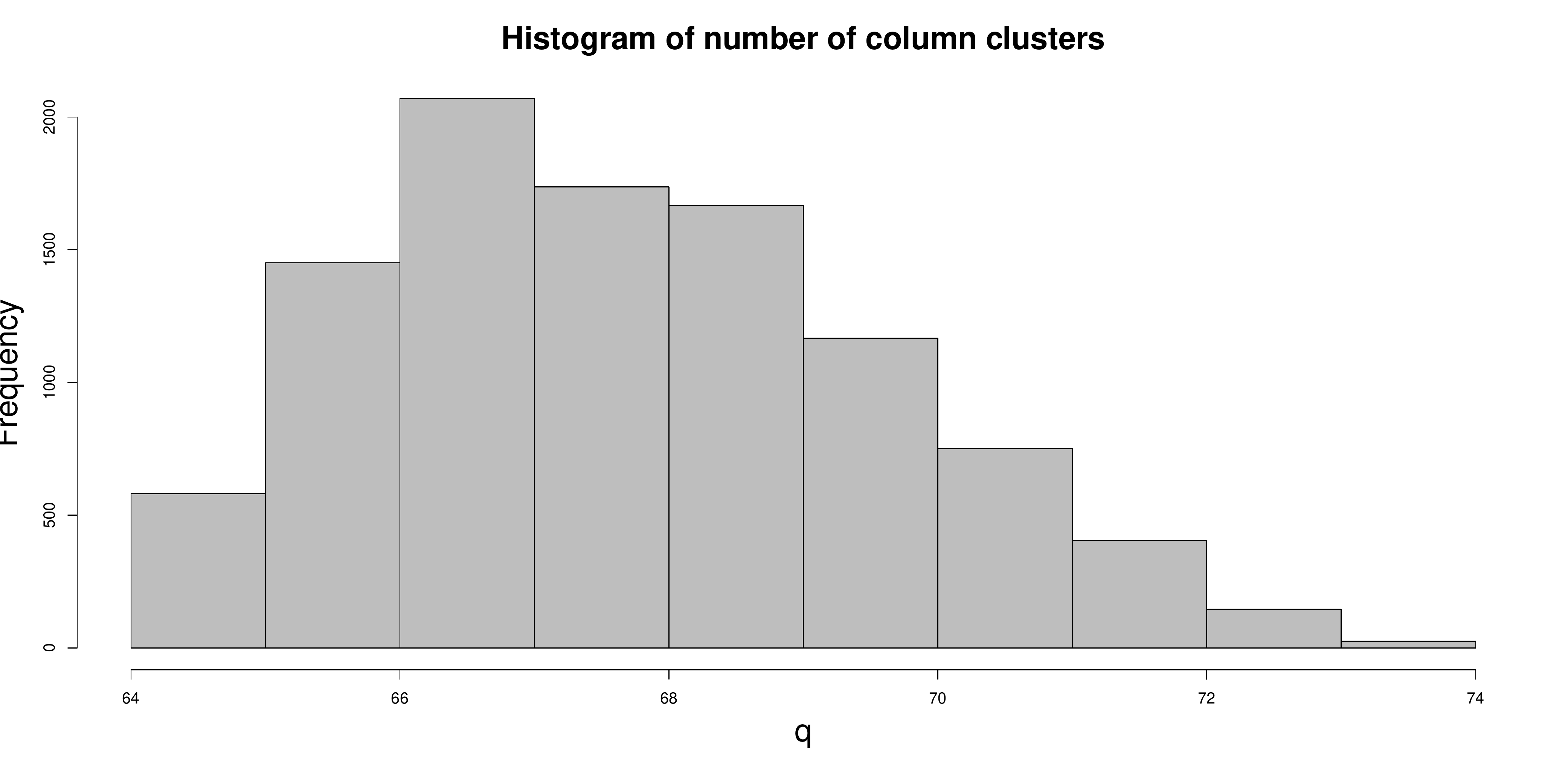} \\
\includegraphics[width=12cm, height =6cm]{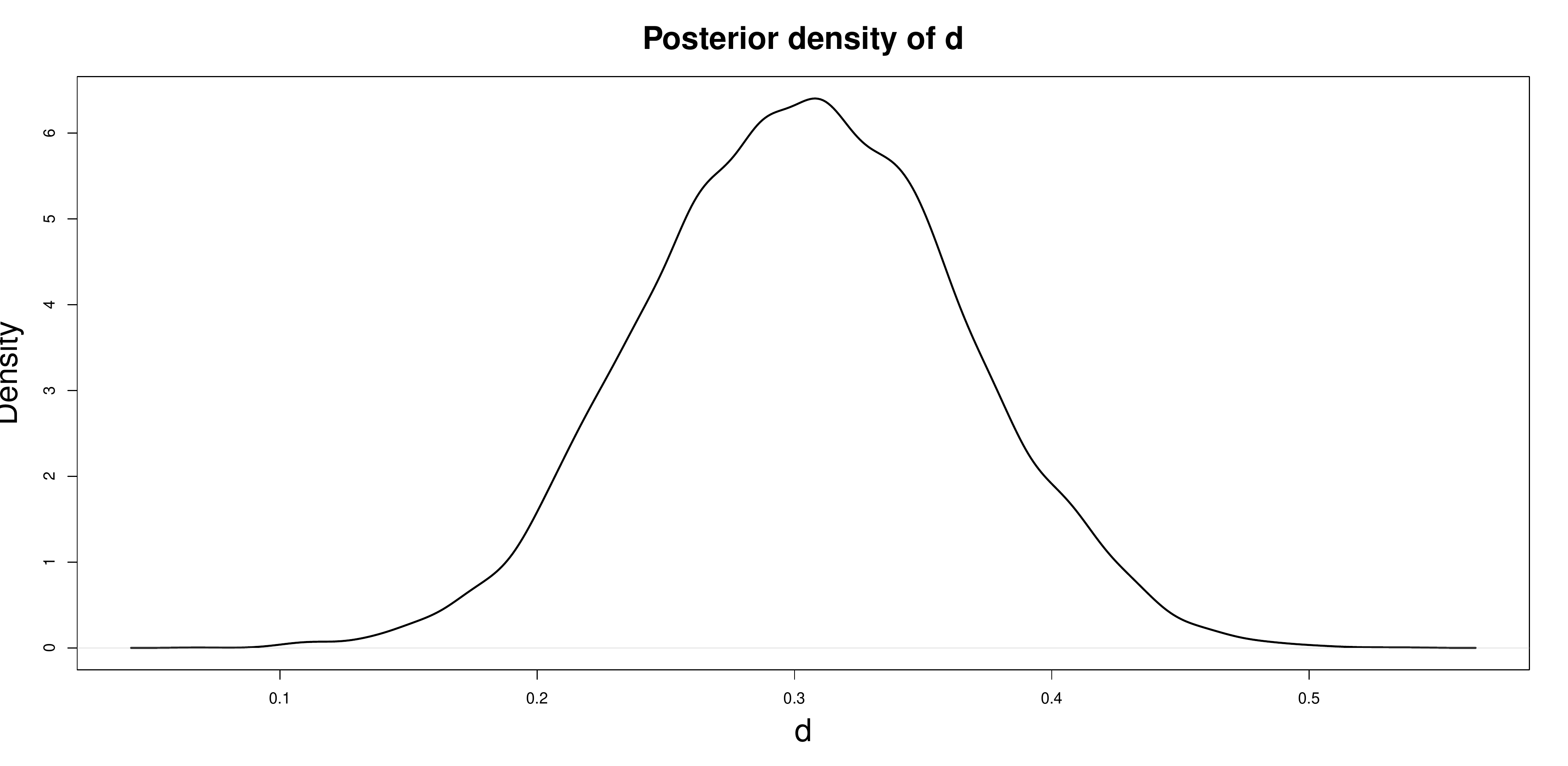}
\caption{The figure on top gives histogram of q and the figure on bottom gives posterior density of d.}
\label{F:postsum}
\end{figure}

We analyzed the merged dataset (GBM data) using our method. Figure \ref{F:postsum} gives the posterior summary of q (number of clusters) and d (discount parameter). The estimated number of covariate cluster turned out to be $\hat{q}$ = 70. 

\begin{figure}[h]
\includegraphics[width=12cm,height=8cm]{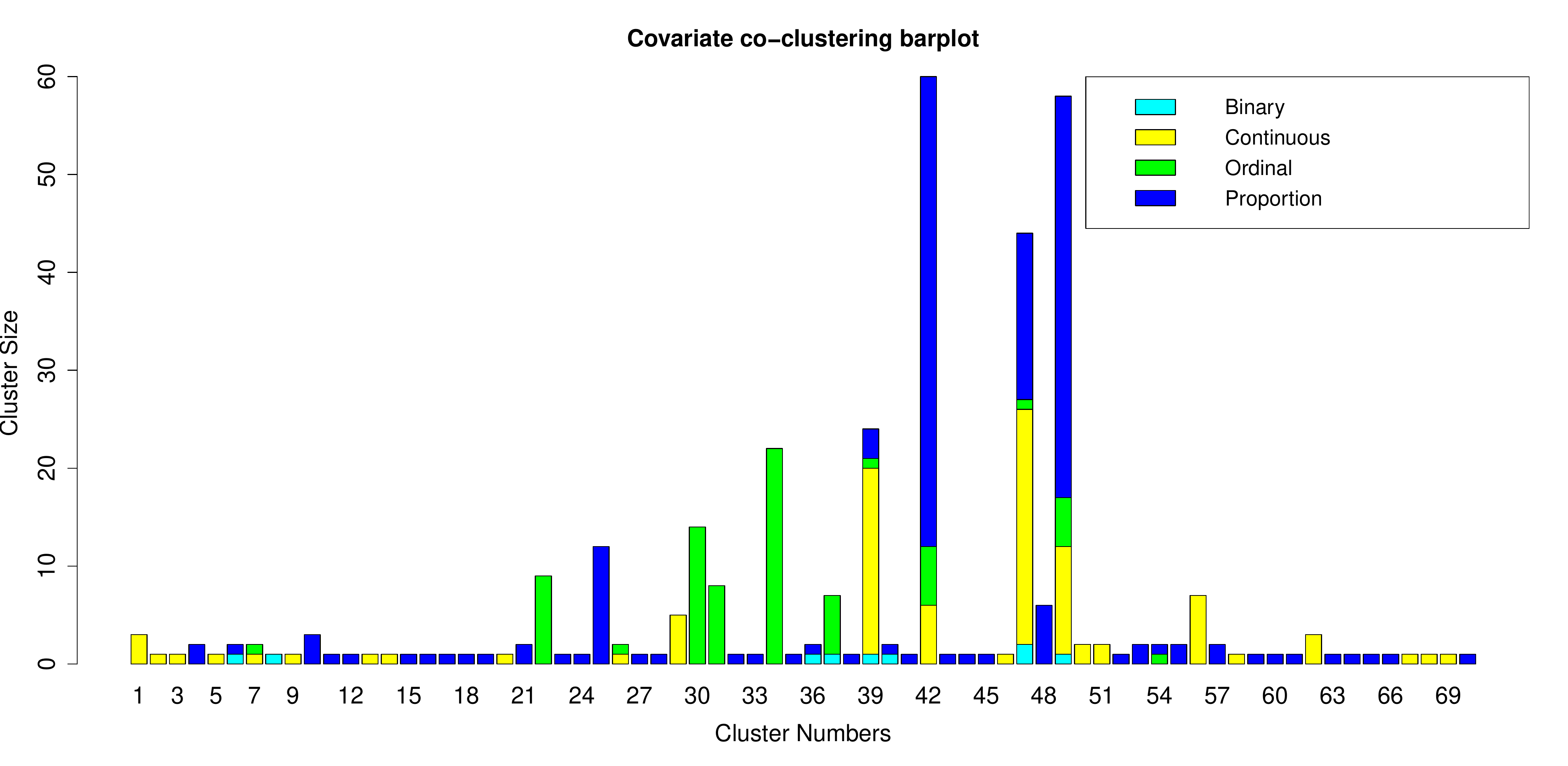}
\caption{Barplot gives cluster membership of covariates across data types.}
\label{F:bargraph}
\end{figure}

\begin{figure}[h]
\centering
\includegraphics[width=12cm, height = 6cm]{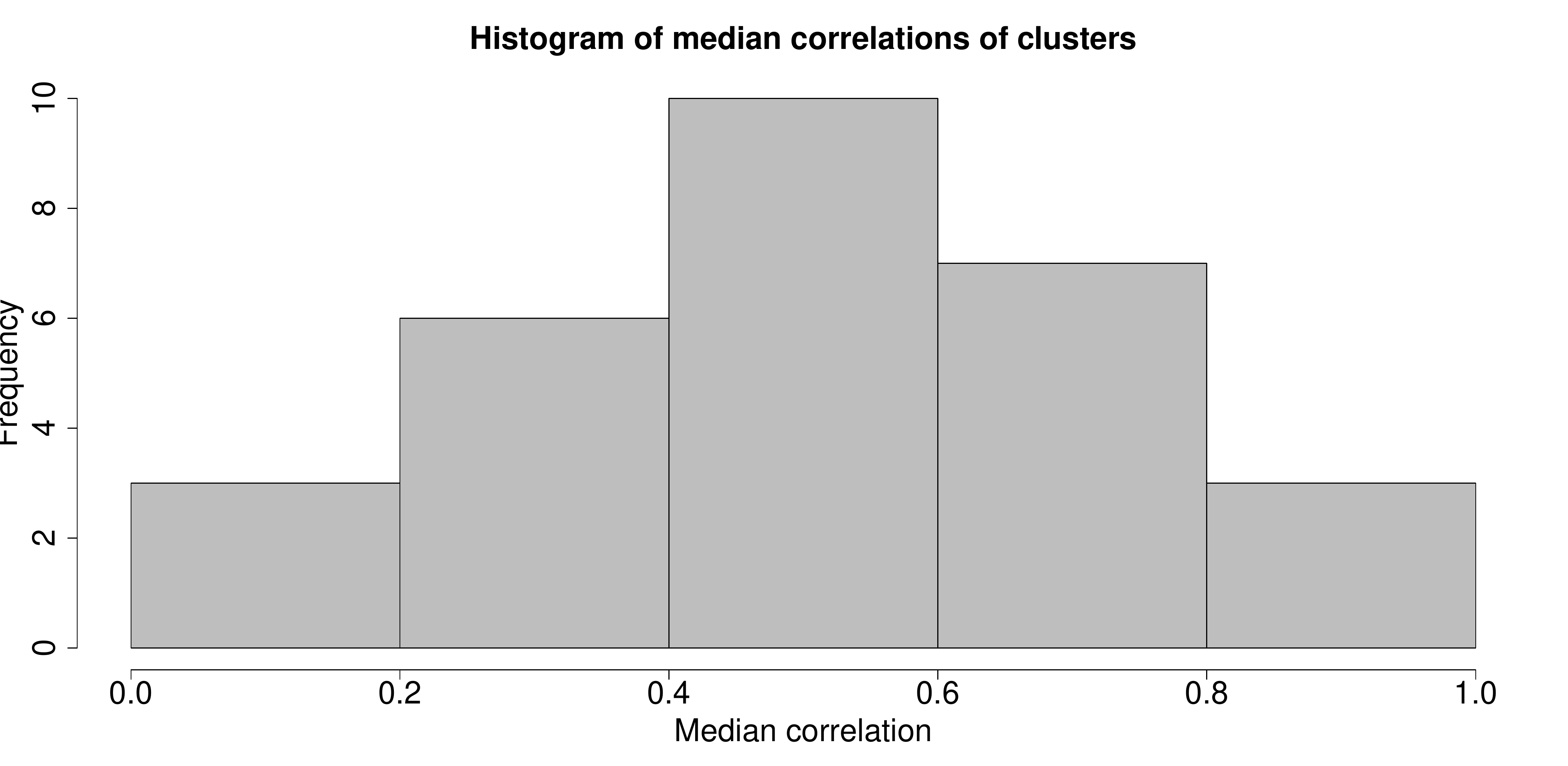} \\
\includegraphics[width =12cm, height =6cm]{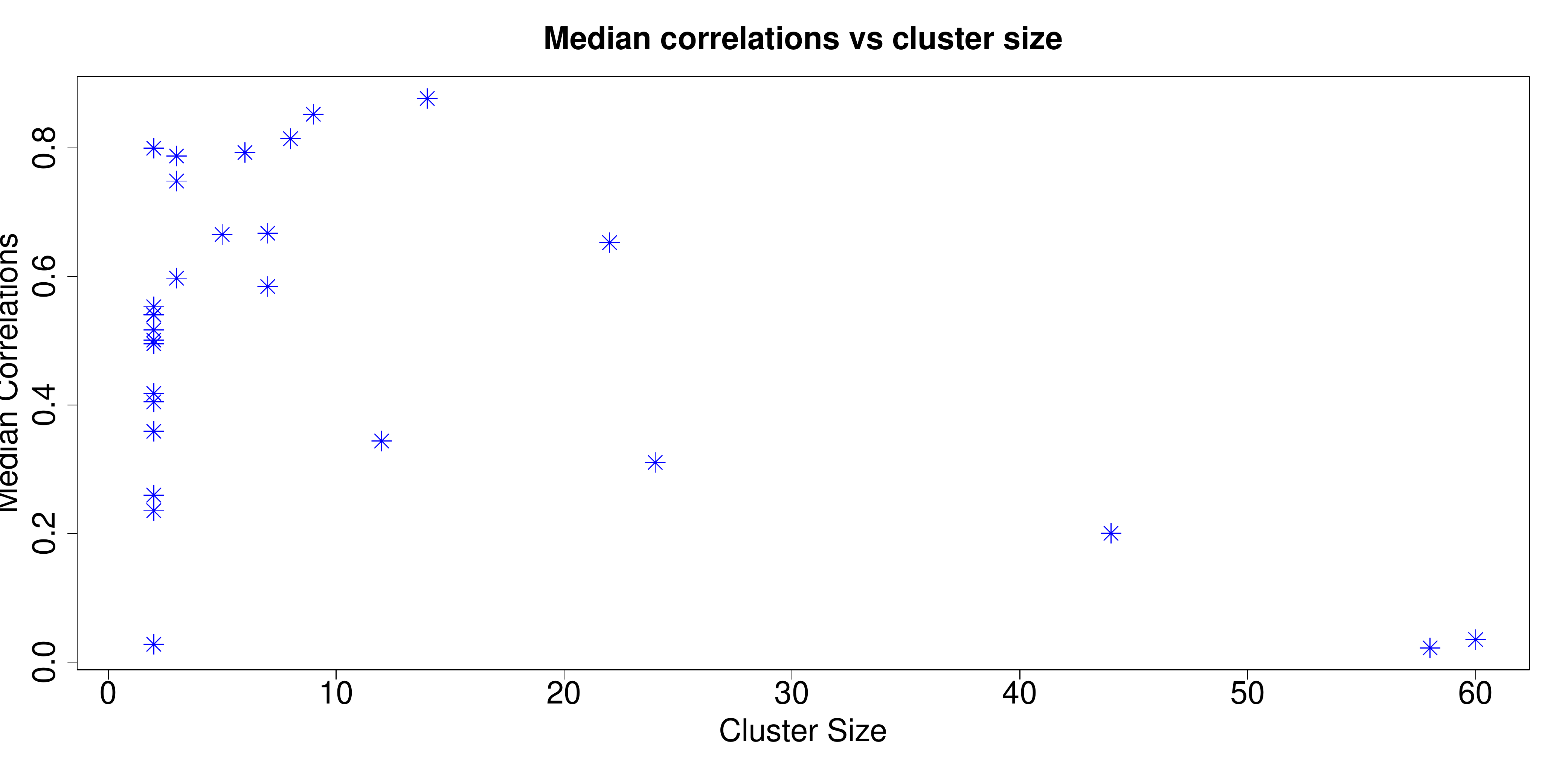}
\caption{ The figure on top gives  histogram of median correlation, and the figure on bottom gives plots median correlation against cluster size.}
\label{F:medclust}
\end{figure}

Figure \ref{F:bargraph} shows that majority of the covariates are allocated into a fewer clusters of moderate/large size, and a large number of smaller clusters. This is similar to sparsity structure imposed by clustering under PDP process.
The bargraph also gives a visual insight into shared cluster membership of covariates across data types. The bargraph color codes the proportion covariate as blue, ordinal covariates as green, continuous covariates as yellow, and binary covariates as cyan. It is worth noting that the common denominator in majority of the shared clusters (across data types) is continuous data type, which corresponds to the gene expression data. The existing knowledge of biological mechanism confirms this finding, and moreover it shows that the above covariate partition captures useful biological information \citep{savage}.
Finally, we see there are non-singleton clusters which are of the same data type. This phenomenon is confirmed by the existing knowledge of system biology, see \citep{ibag}.


The effectiveness of our model could be further demonstrated as follows. For each of the estimated non-singleton covariate clusters, we compute the correlation between its member covariates in the original data. The median correlation gives a summary of associations in these clusters. The median correlations of estimated covariate clusters are plotted in Figure \ref{F:medclust}. We note that the median correlations are all positive. This strongly suggests that our method has found groups of biomarkers which share a similar pattern in the original dataset. We also plot the median correlation against the cluster size in Figure \ref{F:medclust}, which confirms that the median correlation of estimated clusters is moderate despite the cluster size.

Furthermore, we looked at the transformed covariates which were allocated into covariate clusters with more than 10 covariates (biomarkers). We compared the heatmap of such transformed covariates with the heatmap of corresponding original covariates in Figure \ref{F:heatmap}. The gain from our method is apparent, as the heatmap of transformed covariates is comparatively even, whereas the heatmap of the original covariates is  stark, in comparison. We also partition the heatmap of the transformed covariate into different clusters on the basis of estimated cluster membership, using solid black lines, see Figure \ref{F:heatmap}. We note that the heatmap of the transformed covariates within a cluster is quite similar.

For biological relevance of our findings, we looked up biomarkers on the same gene which share the estimated cluster. This could 
explain the role of biological association on same gene (across data type) play in tumor growth. Subsequently, we verified our findings about these genes and their interactions by cross-referencing them in cancer literature. We found that genes associated with GBM cancer were mutated, amplified or methylated to have a positive associated with gene expression. For instance, copy number alteration of EGFR \citep{egfr2}, MDM4 \citep{mdm4}, mutation of TP53 gene \citep{p53}, and methylation of CKDN2B gene \citep{cdkn2b} and PTEN gene \citep{pten} at certain sites
were associated with corresponding gene expression on respective genes. We corroborated these findings by looking up in the cancer literature. Our findings gives us an insight into the underlying biology and how
these associations play a part in tumor growth.
A detailed analysis of the mechanistic interpretation of these genes is given in Appendix A.5.

\begin{figure}[h]
\centering
\includegraphics[width=12 cm, height=6cm]{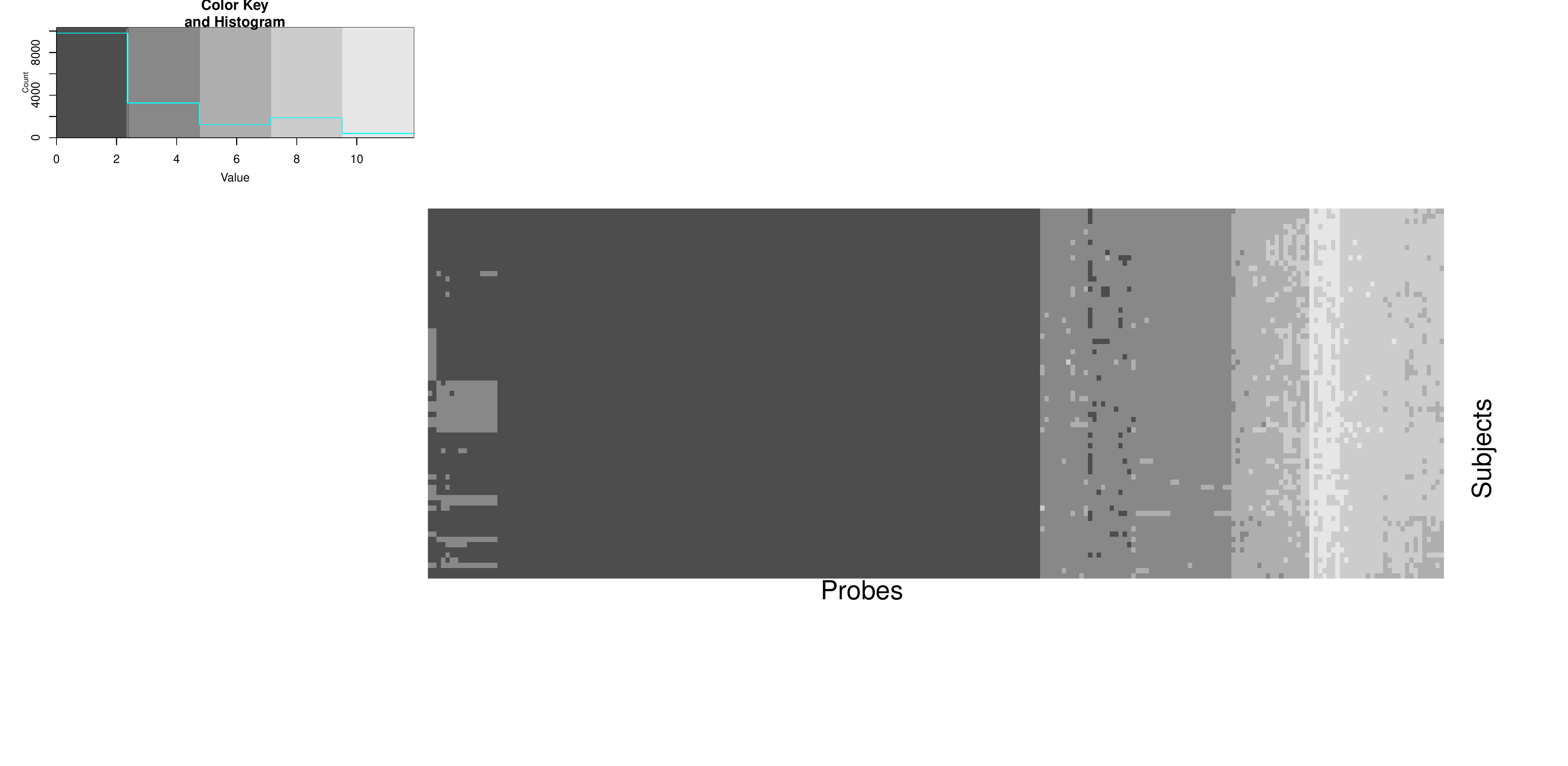}
\includegraphics[width=12cm, height =6cm]{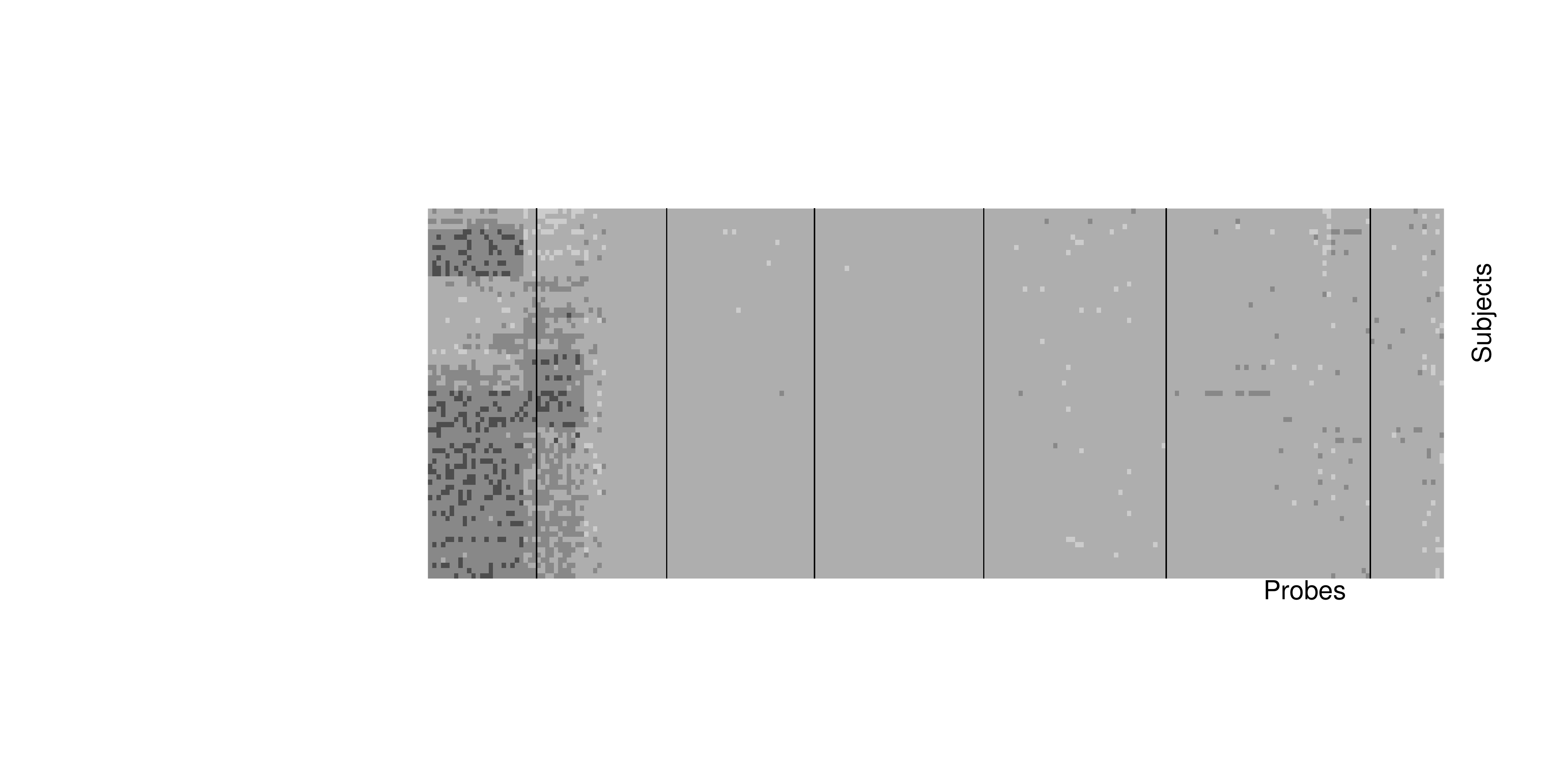}
\caption{ Heatmap of GBM covariates that were assigned to column clusters with more than 10 members. The top panel corresponds to original covariates and the bottom panel corresponds to transformed covariates. The vertical lines in the bottom panel groups covariate clusters. The color key is provided at the top of the panel.}
\label{F:heatmap}
\end{figure}

\section{Conclusion/Discussion}\label{S:conclusion}
Using data augmentation approach, we extended \citet{guhaveera}'s approach to mixed datasets. \emph{Gen-VariScan} is a flexible technique for clustering of high-dimensional mixed datasets. It groups mixed covariates
into small number of clusters, which consists of similar covariates.
We provide theoretical justification for the data augmentation approach, and also prove that our method can detect true co-clustering of covariates. We also demonstrate the effectiveness of our model through simulation and real data analysis. The proposed method outperforms existing approaches for the mixed dataset. In real data analysis, we identified several biological association and interaction, which has known implications in development and progression of cancer. As a byproduct of our work, we also propose a working value approach for beta regression for constant dispersion.

Potentially, one could use our approach to perform survival regression, or find subtypes among patients. It would be worth investigating the gains from such an approach. Another area of investigation is to study the convergence rates of the co-clustering
of covariates. The convergence rate can help us foresee the performance of our method on a large dataset. Finally, our beta regression approach can be used in other settings.

\section*{Supplementary Material}
Supplementary material includes additional details on workings of \emph{Gen-VariScan}, biological interpretations of Benchmark Data Analysis, and proofs of theorems.




\bibliography{reference}

\section*{Acknowledgements}
We would like to thank Dr Subharup Guha for sharing his work on \citet{guhaveera}, without which the current work wouldn't have been possible. This project was partly supported by NSF Grant 1416948. 
\end{document}